\documentclass{scrartcl}
\newcommand{\thebibstyle}{abbrvnat}

\usepackage{amsmath}
\usepackage{amssymb}
\usepackage{amsthm}
\usepackage{bm}
\usepackage{color}
\usepackage{graphicx}
\usepackage{hyperref}
\usepackage[capitalise]{cleveref} 
\usepackage{natbib}
\usepackage{xr}
\usepackage{enumitem}
\usepackage[title]{appendix}
\usepackage{subcaption}

\newtheorem{theorem}{Theorem}[section]
\newtheorem{proposition}[theorem]{Proposition}

\newtheorem{lemma}[theorem]{Lemma}
\theoremstyle{definition}
\newtheorem{definition}[theorem]{Definition}
\newtheorem{example}[theorem]{Example}
\newtheorem{remark}[theorem]{Remark}

\usepackage{adjustbox}
\usepackage{array}
\usepackage{multirow}

\newcommand{\abs}[1]{|#1|}

\newcommand{\norm}[1]{\|#1\|}
\newcommand{\set}[1]{\{#1\}}
\newcommand{\quark}{\setbox0\hbox{$x$}\hbox to\wd0{\hss$\cdot$\hss}}
\newcommand{\reals}{\mathbb{R}}

\newcommand{\wrt}{\textup{d}}
\newcommand{\mat}[1]{\mathbf{#1}}

\newcommand{\itermethod}{\mathcal{I}}
\newcommand\numberthis{\addtocounter{equation}{1}\tag{\theequation}}

\newcolumntype{R}[2]{%
    >{\adjustbox{angle=#1,lap=\width-(#2)}\bgroup}%
    l%
    <{\egroup}%
}

\newcommand{\xtrue}{\mathbf{x}}
\newcommand{\xdummyvec}{\mathbf{v}}
\newcommand{\xdummycomponent}{v}

\DeclareMathOperator{\range}{range}

\newcommand{\Sigmat}{\mathbf{\Sigma}}

\newcommand{\Omegamat}{\mathbf{\Omega}}

\newcommand{\Bmat}{\mathbf{B}}

\newcommand{\Gmat}{\mathbf{G}}

\newcommand{\Imat}{\mathbf{I}}

\newcommand{\Qmat}{\mathbf{Q}}

\newcommand{\Tmat}{\mathbf{T}}

\newcommand{\Wmat}{\mathbf{W}}
\newcommand{\Ymat}{\mathbf{Y}}

\newcommand{\xvec}{\mathbf{x}}

\newcommand{\zerovec}{\mathbf{0}}

\DeclareMathOperator{\N}{\mathcal{N}}

\title{Probabilistic Iterative Methods for Linear Systems}
\author{Jon Cockayne\thanks{The Alan Turing Institute (jcockayne@turing.ac.uk)}\and Ilse C.F. Ipsen\thanks{North Carolina State University (ipsen@ncsu.edu)}\and Chris J. Oates\thanks{Newcastle University and The Alan Turing Institute (chris.oates@newcastle.ac.uk)}\and Tim W. Reid\thanks{North Carolina State University (twreid@ncsu.edu)}}

\begin{document}

\maketitle

\begin{abstract}

This paper presents a probabilistic perspective on iterative methods for approximating the solution $\xtrue \in \mathbb{R}^d$ of a nonsingular linear system $\mat{A} \xtrue = \mathbf{b}$.
Classically, an iterative method produces a sequence $\mathbf{x}_m$ of approximations that converge to $\xtrue$ in $\mathbb{R}^d$.
Our approach, instead, lifts a standard iterative method to act on the set of probability distributions, $\mathcal{P}(\mathbb{R}^d)$, outputting a sequence of probability distributions $\mu_m \in \mathcal{P}(\mathbb{R}^d)$.
The output of a probabilistic iterative method can provide both a ``best guess'' for $\xtrue$, for example by taking the mean of $\mu_m$, and also probabilistic uncertainty quantification for the value of $\xtrue$ when it has not been exactly determined.
A comprehensive theoretical treatment is presented in the case of a stationary linear iterative method, where we characterise both the rate of contraction of $\mu_m$ to an atomic measure on $\xtrue$ and the nature of the uncertainty quantification being provided.
We conclude with an empirical illustration that highlights the potential for probabilistic iterative methods to provide insight into solution uncertainty.
\end{abstract}


\section{Introduction}

The focus of this paper is on the numerical solution of a linear systems of equations
\begin{equation} \label{eq:the_system}
	\mat{A} \xtrue = \mathbf{b}
\end{equation}
where $\mat{A} \in \reals^{d \times d}$ is a given non-singular matrix, $\mathbf{b} \in \reals^d$ is a non-zero vector and $\xtrue \in \reals^d$ is an unknown vector to be computed.
The problem of solving linear systems is central to scientific computation \citep[p103]{Golub2013}.
Solvers can broadly be categorized as either \emph{direct}, meaning they compute $\xtrue$ by factorizing the matrix $\mat{A}$, or as \emph{iterative}, meaning they output a sequence of approximations to $\xtrue$.
The focus of the present paper is on a probabilistic version of iterative methods.

There exist a wide variety of iterative methods, with the two main classes being the stationary iterative methods \citep{Young1971}, such as Richardson's method and Jacobi's method, and Krylov subspace methods \citep{Liesen2012} such as the conjugate gradient method \citep[CG;][]{Hestenes1952}.
In each case, the output of an iterative method is a sequence $\mathbf{x}_m$ of approximations to $\xtrue$, that one hopes will converge to $\xtrue$ as $m$ is increased.
In practice the error $\mathbf{e}_m = \xtrue - \mathbf{x}_m$ is unknown but can be estimated.
Error estimation for linear systems has a long history, with \citet{vonNeumann1947} among the earliest works in this now vast literature.
For CG applied to a symmetric positive definite matrix $\mat{A}$, one typically estimates a bound for the $\mat{A}$-norm of the error $\norm{\mathbf{e}_m}_{\mat{A}} = \sqrt{\mathbf{e}_m^\top \mat{A} \mathbf{e}_m}$.
Estimates such as this may be of limited utility for three reasons: they are often conservative, they may be complicated to compute and, being a scalar-valued summary, they cannot capture all of the structure that may be present in the error $\mathbf{e}_m$.

The purpose of this paper is to lift standard iterative methods into probability space, replacing iterates $\mathbf{x}_m \in \reals^d$ with iterates $\mu_m \in \mathcal{P}(\reals^d)$, where $\mathcal{P}(\reals^d)$ denotes the set of probability measures on $\reals^d$.
The output of such a method then simultaneously provides an approximation to $\xtrue$, for example by taking the mean of $\mu_m$, and probabilistic error assessment.
To motivate why such a method may be useful, suppose that the value of $\xtrue$ is the input to some further computation, denoted abstractly as $F(\xtrue)$ for $F : \reals^d \to \reals$, and suppose that one wishes to characterise the error $F(\xtrue) - F(\mathbf{x}_m)$ in replacing the unknown $\xtrue$ with the numerical approximation $\mathbf{x}_m$.
It is not trivial to transfer a bound on a derived quantity such as $\norm{\mathbf{e}_m}_{\mat{A}}$ into a practically useful estimate of this error, particularly when $F$ is not analytically tractable.
For example, if $F$ depends only on a subset of the entries of $\xtrue$ for which the iterative method converges rapidly, while the other entries converge slowly, a bound on $F(\xtrue) - F(\mathbf{x}_m)$ that is a function only of $\norm{\mathbf{e}_m}_{\mat{A}}$ can be too conservative to be useful.
In contrast, a probabilistic representation $\mu_m$ of uncertainty regarding $\xtrue$ can be directly propagated through $F$ by repeatedly sampling $\bm{X} \sim \mu_m$ and computing $F(\bm{X})$.
The resulting probability distribution provides probabilistic  \emph{uncertainty quantification} (UQ) for the unknown quantity of interest $F(\xtrue)$, and may not suffer the same degree of conservatism of the norm-based estimators that we briefly described.

The methods described herein can be viewed as \textit{probabilistic numerical methods} \citep[PNM;][]{Larkin1972,Diaconis1988,Hennig2015a,Cockayne2019,oates2019modern}.
PNMs for linear systems are numerical methods that take as input the quantities $\mat{A}$ and $\mathbf{b}$, together with an \textit{initial distribution} $\mu_0 \in \mathcal{P}(\mathbb{R}^d)$, and return a probability distribution $\mu_m \in \mathcal{P}(\mathbb{R}^d)$ as their output.
The role of $\mu_0$ is to encode any \emph{a priori} information that can be provided to the PNM. 
This is achieved by assigning probability mass to subsets of $\mathbb{R}^d$ in which $\xtrue$ is believed to be located, prior to any computations being performed.
This information may be elicited from a domain expert or obtained in an objective manner, for instance by performing additional computations related to the numerical task.
While such applications to numerics have a different flavour to traditional applications of UQ \citep[e.g.][]{Smith2014}, the use of probabilities to describe uncertainty is philosophically similar; see further discussion in \citet{Hennig2015a} and \citet{Cockayne2019}.

\subsection{Related Work} \label{sec:related}

There has been recent interest in the construction of PNM for the solution of \cref{eq:the_system}, with contributions in \cite{Hennig2015a,Bartels2016,Bartels2019,Cockayne2019b,Reid2020,Wenger2020}.
With the exception of \citet{Bartels2019}, these works have predominantly focused on replicating CG, and so a positive-definite $\mat{A}$ is assumed.
Each of these works constructed a PNM in the Bayesian statistical framework, where the distribution $\mu_0$ has the interpretation of a \emph{prior} posited over some quantity related to \cref{eq:the_system} at the outset, and this distribution is updated based on the limited computations that are performed.
The updating is achieved using Bayes' theorem and the result is a \emph{posterior} or \emph{conditional} distribution $\mu_m$ that forms the output of the method; it is a distribution over the unknown $\xtrue$ that quantifies uncertainty given the limited computation performed.
In \cite{Hennig2015a,Bartels2016,Wenger2020} the prior was placed on the entries of $\mat{A}^{-1}$ (or jointly on $\mat{A}$ and $\mat{A}^{-1}$), while in \cite{Bartels2019,Cockayne2019b} the prior was placed directly on the unknown solution of \cref{eq:the_system}.
In each case, computation consisted of projecting \cref{eq:the_system} against a set of search directions $\mathbf{s}_i$, $i=1,\dots,m$ (i.e.\ by computing $\mathbf{s}_i^\top \mat{A} \xtrue = \mathbf{s}_i^\top \mathbf{b}$) and the output of the PNM was a distribution that contracts to a point mass at $\xtrue $ in an appropriate computational limit.

Each of these methods exploited conjugacy of Gaussian distributions under linear transformations to condition on the linear information provided by the pairs $(\mathbf{s}_i,\mathbf{s}_i^\top \mathbf{b})$, $i = 1,\dots,m$.
This conditioning is justified \emph{only} when the search directions $\mathbf{s}_i$ are not themselves dependent on $\xtrue $, the solution of \cref{eq:the_system}.
However, in practice these authors advocated the use of search directions generated using a Lanczos-style recursion \citep[Section 2.4]{Liesen2012}, meaning that the $\mathbf{s}_i$ depend on $\xtrue $ via $\mathbf{b}$ and the required assumption is violated.
As remarked in \cite{Bartels2019,Cockayne2019b}, this violation leads to PNM that are neither Bayesian nor \textit{calibrated}, with the latter understood to mean that the ``width'' of the probability distribution $\mu_m$ produced by the PNM can be a gross over-estimate of the actual error, as quantified by the difference between the mean of $\mu_m$ and $\xtrue $.
\citet{Reid2020} addressed this deficiency by constructing a prior which corrects for the over-confidence in an empirical Bayesian fashion, though with such a prescribed prior it is difficult for other problem-specific information to be incorporated.
It therefore remains an open problem to develop a PNM for the solution of \cref{eq:the_system} that allows a generic initial distribution $\mu_0 \in \mathcal{P}(\mathbb{R}^d)$ to be used and ensures the distributional output $\mu_m \in \mathcal{P}(\mathbb{R}^d)$ of the PNM is calibrated.

\subsection{Contributions}

This paper adopts a different strategy to the aforementioned work.
Instead of applying Bayes' theorem, we first posit an initial distribution $\mu_0$ and iteratively update this distribution using a transformation derived from a standard iterative method for the solution of \cref{eq:the_system}.
The initial distribution $\mu_0$ is loosely analogous to the prior in a Bayesian approach, but since no analogue of the Bayesian update occurs these methods are not Bayesian and we refrain from using the terms \textit{prior} and \textit{posterior} in this work.
We thus refer to $\mu_m$ as a \emph{belief distribution}, following the contemporary literature on generalised Bayesian inference \citep{Bissiri2016}.
In departing from an established statistical paradigm one is required to justify, mathematically, the sense in which the uncertainty quantification provided by $\mu_m$ is meaningful.
For this purpose we leverage the recent work of \citet{Oates2020}, who argued that non-Bayesian procedures can be justified if they are \emph{calibrated}, meaning that the unknown true solution $\xtrue $ is indistinguishable in a certain, statistical sense, from any other sample drawn independently from $\mu_m$.
The contributions of this paper are therefore as follows:
\begin{itemize}
	\item We introduce \emph{probabilistic iterative methods}, a class of PNM derived from iterative methods for solving linear systems such as \cref{eq:the_system}.
	These methods can be interpreted as a lifting of standard iterative methods into probability space, and are equivalent to randomising the initial iterate in a standard iterative method.
	\item A detailed theoretical analysis of the convergence properties of these new PNM is conducted for the class of linear stationary iterative methods, in which the next iterate is obtained by an affine transformation of the previous iterate.
	We prove that in this case the iterates produced are \emph{strongly calibrated} in the sense of \citet{Oates2020} and hence provide meaningful uncertainty quantification despite not existing in the Bayesian paradigm.
	\item An empirical assessment is performed to determine whether or not probabilistic iterative methods based on more complex iterative methods, such as Krylov methods, are calibrated.
	\item A simulation study is conducted to analyse the performance of probabilistic iterative methods in a toy regression context.
	Here we examine the convergence and calibration of both linear and nonlinear probabilistic iterative methods, and highlight how their output may be used to gain insight into the impact of numerical uncertainty in the context of the regression task.
\end{itemize}

\subsection{Structure of the Paper}

In \cref{sec:stationary_iterative_introduction} we introduce iterative methods for linear systems and describe how these may be lifted into algorithms that operate on probability space.
Theoretical results concerning the convergence and calibration of a class of analytically tractable probabilistic iterative methods are presented in \cref{sec:linear_probabilistic_methods}, and in \cref{sec:extensions} we consider the general case, presenting a statistical test that can be used to assess whether the output from a probabilistic iterative method is calibrated.
In \cref{sec:experiments} we apply probabilistic iterative methods to solve a linear system arising in a regression problem.
Lastly, in \cref{sec:conclusion} we discuss the results presented and the outlook for this new class of methods.

\subsection{Notation} \label{sec:notation}

Here the notation for the paper is established.
We will work in the measurable space $(\reals^d, \mathcal{B}(\reals^d))$ where $\mathcal{B}(\reals^d)$ is the standard Borel sigma-algebra for $\reals^d$.
Let $\mathcal{P}(\reals^d)$ denote the set of all probability measures on $(\reals^d, \mathcal{B}(\reals^d))$.
Bold lower-case roman letters (e.g.\ $\xdummyvec$) will be used to denote vectors in $\reals^d$ and bold capital roman letters to denote matrices in $\reals^{d\times d}$ (e.g.\ $\mat{M}$).
Bold capital italic letters will denote random variables on $\reals^d$ (e.g.\ $\bm{X}$) and lower-case Greek letters (e.g. $\mu$) will be used to denote elements of $\mathcal{P}(\reals^d)$.

Throughout it will be assumed that $\norm{\quark}$ is a fixed but arbitrary norm on $\reals^d$.
One important example is the vector $p$-norm, given by
\begin{equation*}
\norm{\xdummyvec}_p = \left(\sum_{i=1}^d \abs{\xdummycomponent_i}^p\right)^{\frac{1}{p}} ,
\end{equation*}
though we note that many of the results presented herein do not assume any particular norm, and where a specific norm \emph{is} required this will be emphasised.
This notation will also be used for the induced norm on $\reals^{d \times d}$, given by
\begin{equation*}
\norm{\mat{M}} = \sup_{\norm{\mathbf{x}} = 1} \norm{\mat{M} \mathbf{x}} .
\end{equation*}
Recall that all induced norms are sub-multiplicative, meaning that $\|\mat{M}\xdummyvec\| \leq \|\mat{M}\| \|\xdummyvec\|$.
Let $\rho(\mat{M})$ denote the spectral radius of $\mat{M}$, let $\mat{M}^\dagger$ denote the Moore-Penrose pseudo-inverse of $\mat{M}$, let $\textup{range}(\mat{M})$ denote its range and $\textup{ker}(\mat{M})$ its kernel or null space.
For a symmetric matrix $\mat{M}$, let $\lambda_{\min}(\mat{M})$ and $\lambda_{\max}(\mat{M})$ denote the smallest and largest eigenvalue of $\mat{M}$.
For a positive-definite matrix $\mat{M}$ we define the weighted norm $\norm{\xdummyvec}_{\mat{M}} = (\xdummyvec^\top \mat{M} \xdummyvec)^{1/2}$.
Let $\mat{M}^{1/2}$ denote a matrix for which $\mat{M} = (\mat{M}^{1/2})^\top \mat{M}^{1/2}$.
Note that this is not the typical notion of a square root, in that it will \emph{not} be required that $(\mat{M}^{1/2})^\top = \mat{M}^{1/2}$.


For a measurable map $S : \reals^d \to \reals^d$ and a set $B \subset \reals^d$, $S^{-1}[B]$ will be used to denote the preimage of $B$ under $S$, i.e.
\begin{equation*}
	S^{-1}[B] = \set{\xdummyvec \in \reals^d \text{ s.t. } S(\xdummyvec) \in B}.
\end{equation*}
For a distribution $\mu \in \mathcal{P}(\reals^d)$, recall that the pushforward distribution $S_\# \mu$ is the element of $\mathcal{P}(\mathbb{R}^d)$ defined as $(S_\#\mu)(B) = \mu(S^{-1}[B])$ for each $B \in \mathcal{B}(\reals^d)$.
The notation $\mathcal{N}(\xdummyvec, \mat{\Sigma})$ will be used to denote the multivariate Gaussian distribution with mean $\xdummyvec$ and positive semi-definite covariance $\mat{\Sigma}$.
The notation $\chi^2_d$ will denote the chi-squared distribution with $d \in \mathbb{N}$ degrees of freedom.
Recall that if $\bm{X} \sim \mathcal{N}(\mathbf{0}, \mat{I}_d)$ then $\norm{\bm{X}}_2^2 \sim \chi^2_d$.

\section{Probabilistic Iterative Methods} \label{sec:stationary_iterative_introduction}
In this section we start by recalling standard iterative methods, using the taxonomy of \cite{Young1971}, before then presenting our new concept of a probabilistic iterative method.

\subsection{Iterative Methods}

A general iterative method $\itermethod$ is defined \citep[Section 3.1]{Young1971} as a sequence of maps $\mathcal{I} = (P_m)_{m\geq 1}$, for which $\mathbf{x}_m = P_m(\mathbf{x}_0, \dots, \mathbf{x}_{m-1}; \mat{A}, \mathbf{b})$.
The notation $\itermethod(\mat{A}, \mathbf{b})$ will occasionally be used to make the dependence of the iterative method on $\mat{A}$ and $\mathbf{b}$ explicit.
The iterative method $\itermethod$ is said to be \emph{linear} if each $P_m$ is linear in $\mathbf{x}_0,\dots,\mathbf{x}_{m-1}$.
It is said to be \emph{of degree $s$} if for all $m \geq s$ we have that $P_m$ depends only on the $s$ previous iterates, i.e.\ $P_m(\mathbf{x}_0,\dots,\mathbf{x}_{m-1}; \mat{A}, \mathbf{b}) = P_m(\mathbf{x}_{m-s},\dots,\mathbf{x}_{m-1}; \mat{A}, \mathbf{b})$.
Lastly, the method is said to be \emph{stationary} if the maps $P_m$ are independent of $m$.
In what follows we tend to suppress dependence of $\itermethod$ and the $P_m$ on $\mat{A}$ and $\mathbf{b}$ to reduce notational overhead.

Many of the most widely used iterative methods can be expressed as methods of degree $s = 1$.
For simplicity, we present the majority of the material in this paper in these terms, though the core ideas readily generalise to higher degree methods as will be explored in \cref{sec:extensions,sec:experiments}.
Any iterative method $\itermethod$ of degree $s = 1$ implies a map $P^m$ that acts only on the first iterate $\mathbf{x}_0$ to produce iterate $\mathbf{x}_m$, as follows:
\begin{equation*}
	P^m(\mathbf{x}_0) = (P_m \circ \cdots \circ P_1)(\mathbf{x}_0) .
\end{equation*}
In such cases each $P_m$ is generally a \emph{contraction map} with fixed point $\xtrue $, i.e.\ $P_m(\xtrue ) = \xtrue $.
Thus, when the iterative method is stationary it amounts to applying a single fixed contraction map to an initial iterate until convergence.

We now present several examples of first degree iterative methods; for each see \citet[Section 3.3]{Young1971}.
These methods are seldom used as linear solvers in contemporary applications, but are still sometimes used in conjunction with other methods \citep[p103]{Saad2003}.

\begin{example}[Stationary Richardson method] \label{ex:richardson}

This method adopts the following iteration
\begin{equation*}
	\mathbf{x}_m = \mathbf{x}_{m-1} + \omega (\mathbf{b} - \mat{A} \mat{x}_{m-1} ), \qquad m \geq 1
\end{equation*}
where $\omega > 0$ is a parameter of the method.
The method is stationary and linear, with each map $P_m$ of the form
\begin{equation} \label{eq:jacobi_map}
	P_m(\xdummyvec) = P(\xdummyvec) = \mat{G} \xdummyvec + \mathbf{f}
\end{equation}
where $\mat{G} = \mat{I}_d - \omega\mat{A}$ and $\mathbf{f} = \omega \mathbf{b}$.

\end{example}

\begin{example}[Jacobi's method] \label{ex:jacobi}

In \emph{Jacobi's method} it is assumed that the diagonal elements of $\mat{A}$ are nonzero.
The iteration takes the form
\begin{equation*}	
	\mat{x}_m =  \mat{D}^{-1} ( \mat{b} - (\mat{A} - \mat{D}) \mat{x}_{m-1}) + \mat{x}_{m-1}, \qquad m \geq 1
\end{equation*}
where $\mat{D} = \textup{diag}(\mat{A})$.
The method is again stationary and linear.
In the notation of \cref{eq:jacobi_map}.
we have that $\mat{G} = \mat{I}_d - \omega \mat{D}^{-1} \mat{A}$ and $\mathbf{f} = \omega \mat{D}^{-1} \mathbf{b}$.

\end{example}

The next method, CG, sees significantly more use, particularly in the solution of large sparse linear systems.
Whereas the above two methods are based on matrix splittings, in CG the solution $\xtrue $ is instead projected into a sequence of \emph{Krylov subspaces} \citep[Section 2.2]{Liesen2012} of increasing dimension.
As a result it is not traditionally viewed within the classification of \citet{Young1971}\footnote{
The discussion in \citet[Section 2.5.7]{Liesen2012} highlights that, when \citet{Young1971} was written, CG was still often considered a direct method owing to its convergence in $m' \leq d$ iterations; its attractive properties as an iterative method were not understood by the community until \citet{Reid1971}, who studied its use as an iterative method for large sparse linear systems.
This likely explains why \citet{Young1971} does not attempt to categorise it within his taxonomy.}.
Nevertheless CG is currently seen as an iterative method and may be categorised within the taxonomy presented above, albeit rather degenerately since CG converges (in exact arithmetic) in a finite number $m' \leq d$ of iterations and so $P_m$ is undefined for $m > m'$.

\begin{example}[Conjugate gradient method] \label{ex:cg}

	In CG the iteration is of the form
	\begin{align*}	
		\mathbf{x}_m &= \mathbf{x}_{m-1} + \alpha_m \mathbf{s}_m \qquad&
		\alpha_m &= \frac{\mathbf{s}_m^\top \mathbf{r}_m}{\mathbf{s}_m^\top \mat{A} \mathbf{s}_m} \\
		\mathbf{s}_{m+1} &= \mathbf{r}_{m} + \beta_{m} \mathbf{s}_{m} \qquad&
		\beta_m &= \frac{\mathbf{r}_m^\top \mathbf{r}_m}{\mathbf{r}_{m-1}^\top \mathbf{r}_{m-1}} 
	\end{align*}
	where the initial direction $\mathbf{s}_0$ is taken to be the initial residual $\mathbf{r}_0$, and we recall that recall that $\mathbf{r}_m = \mathbf{b} - \mat{A} \mathbf{x}_m$.
	From \citet[Algorithm 6.19]{Saad2003}, CG may be expressed as a three-term recurrence.
	Examining this, we see that CG is neither stationary nor linear, and is of second degree.
	Nevertheless in terms of its \emph{implementation}, the algorithm requires only the storage of $\mathbf{x}_{m}$ and $\mathbf{r}_{m}$ to compute $\mathbf{x}_{m+1}$.
\end{example}

\subsection{Lifting to Probability Space}

Now we introduce the central definition of this paper, that of a probabilistic iterative method.
As noted above, the definition is presented in terms of a method of degree $s=1$; extension to higher degree is considered in \cref{sec:linear_generalisations}.

\begin{definition} \label{def:iterative_method}
	Let $\itermethod = (P_m)_{m \geq 1}$ be an iterative method of first degree.
	Then the maps $P_m:\reals^d \to \reals^d$ can be lifted to maps $(P_m)_\# : \mathcal{P}(\reals^d) \to \mathcal{P}(\reals^d)$ operating on elements of $\mathcal{P}(\mathbb{R}^d)$.
	We say that $\itermethod_\# = ((P_m)_\#)_{m \geq 1}$ is a \emph{probabilistic iterative method}.
\end{definition}

Thus probabilistic iterative methods are a class of PNMs that take as input an initial distribution $\mu_0 \in \mathcal{P}(\reals^d)$ and return a sequence of iterates $\mu_m = (P_m)_\#\mu_{m-1}$.
Again we note that $\itermethod_\#$, and therefore $\mu_m$, each formally depend on $\mat{A}$ and $\mathbf{b}$, but this dependence is notationally suppressed.
The distribution $\mu_0$ should be thought of as an initial belief about where the solution $\xtrue $ to the linear system might lie in $\mathbb{R}^d$.
Thus $\mu_0$ has a similar role to the prior distribution in the Bayesian setting.
However, the iterates $\mu_m$ do not arise as a conditional distribution, and so the output from probabilistic iterative methods does not have a classical Bayesian interpretation.
It is therefore crucial to ensure that the UQ provided by the method is meaningful.
Indeed, in contrast to a Bayesian approach, it is straightforward to construct an example showing that the support of $\mu_m$ need not be contained in the support of $\mu_0$.
Thus, even if $\mu_0$ encodes properties of the solution that are expected to hold with probability one (for example, positivity of the elements) $\mu_m$ is not guaranteed to inherit those properties.
This emphasises the need for careful analysis of probabilistic iterative methods, which we present in detail in \cref{sec:calibration} (for stationary linear methods) and \cref{sec:testing_weak} (for general methods).

Compared to earlier attempts to construct PNM for solution of \cref{eq:the_system}, probabilistic iterative methods are significantly easier to implement.
For example, an algorithm for producing a sample from $\mu_m$ is to sample $\bm{X} \sim \mu_0$ and compute $P^m(\bm{X})$.
Thus sampling from the output of a probabilistic iterative method inherits the computational efficiency and stability of the underlying iterative method, only multiplying the cost by the number of samples required.
Conversely, earlier approaches to PNM (which had a Bayesian flavour) generally required new algorithms and corresponding code to be developed, whose numerical stability must then be independently tested and verified\footnote{
This is particularly true of existing PNM for solving linear systems such as those methods discussed in \cref{sec:related}.
The Lanczos-style recursions exploited to construct the search directions of those methods are well known to lead to numerical instabilities in algorithms such as CG, and the impact of this on the posterior covariance matrices computed in those methods has, to our knowledge, not yet been analysed.}.

Our first theoretical result shows that if the classical iterates $\mathbf{x}_m$ converge to the true solution $\xtrue$, then the distributions $\mu_m$ contract to an atomic mass centred on $\xtrue$ under weak regularity conditions on $\mu_0$.

\begin{proposition} \label{thm:iterative_contraction}
	Let $\itermethod$ be an iterative method of first degree for solution of \cref{eq:the_system}.
	Suppose that each $P_m$ has error controlled by the bound
	\begin{equation} \label{eq:generic_bound}
		\norm{\xtrue  - P_m(\mathbf{x}_0)} \leq \varphi(m)\norm{\xtrue  - \mathbf{x}_0}, \qquad m \geq 1
	\end{equation}
	where $\varphi : \mathbb{N} \to \reals$ is some function independent of $\mathbf{x}_0$, such that $\varphi(m) \to 0$ as $m \to \infty$.
	Then for any $k > 0$ and $\delta > 0$, 
	\begin{equation*}
		\mu_m(B_\delta^\mathsf{c}(\xtrue )) \leq \left(\frac{\varphi(m)}{\delta}\right)^k \int_{\reals^d}\norm{\xtrue  - \xdummyvec}^k \; \wrt \mu_0(\xdummyvec) ,
	\end{equation*}
	where $B_\delta(\xtrue )$ represents a $\norm{\quark}$-ball of radius $\delta$ about $\xtrue $, i.e. $B_\delta(\xtrue ) = \set{\xdummyvec \in \reals^d : \norm{\xtrue  - \xdummyvec} < \delta}$, and $B_\delta^\mathsf{c}(\xtrue )$ its complement in $\mathbb{R}^d$.
\end{proposition}


\begin{proof}
For any $k > 0$,
\begin{align*}
	\int_{\reals^d}\norm{\xtrue  - \xdummyvec}^k \; \wrt \mu_m(\xdummyvec) 
	&= \int_{\reals^d} \norm{\xtrue  - P^m(\xdummyvec)}^k \;\wrt \mu_0(\wrt \xdummyvec) \\
	&\leq \varphi(m)^k \int_{\reals^d} \norm{\xtrue  - \xdummyvec}^k \; \wrt \mu_0(\wrt \xdummyvec) 
\end{align*}
where the second line follows from \cref{eq:generic_bound} and extracting terms independent of $\xdummyvec$ from the integral.
Now, recall from Chebyshev's inequality \citep[Lemma 3.1]{Kallenberg2002} we have that
for a measure $\mu$ on $\reals^d$, a $\mu$-measurable function $f : \reals^d \to [0, \infty)$ and scalars $\delta \in [0, \infty)$, $k \in (0,\infty)$   it holds that
\begin{equation*}
	\mu(\set{\xdummyvec \in \reals^d : f(\xdummyvec) \geq \delta}) = \mu(\set{\xdummyvec \in \reals^d : f(\xdummyvec)^k \geq \delta^k}) \leq \frac{1}{\delta^k} \int_{\reals^d} f(\xdummyvec)^k \;\wrt \mu(\xdummyvec)  .
\end{equation*}
Applying this in the present setting with $f(\xdummyvec) = \norm{\xtrue  - \xdummyvec}$ we therefore have
\begin{align*}
	\mu_m(B_\delta^\mathsf{c}(\xtrue )) \leq \left(\frac{\varphi(m)}{\delta} \right)^k \int_{\reals^d}\norm{\xtrue  - \xdummyvec}^k \; \wrt \mu_0(\xdummyvec)
\end{align*}
as required.
\end{proof}

Thus the probability mass assigned by $\mu_m$ to the region outside of a ball $B_\delta(\xtrue )$ centred on the true solution $\xtrue $ vanishes as $m \rightarrow \infty$.
Moreover, and again asymptotically as $m \rightarrow \infty$, the probability mass outside $B_\delta(\xtrue )$ vanishes more rapidly when high-order moments of $\mu_0$ exist (i.e.\ for large $k$).
However, \Cref{thm:iterative_contraction} does not imply that the UQ provided by $\mu_m$ is meaningful, or even that $\xtrue $ is in the support of $\mu_m$.
For the UQ to be meaningful further assumptions are required on $\itermethod$, such as those made in \cref{sec:calibration}.

\section{Linear Probabilistic Iterative Methods} \label{sec:linear_probabilistic_methods}

In this section we restrict attention to \emph{linear, stationary} iterative methods of first degree, as a richer set of theoretical results can be developed for this restricted set of methods.
In \cref{sec:linear_probabilistic} we recall some classical results and describe how the probabilistic iterates $\mu_m$ can be exactly computed when $\mu_0$ is Gaussian.
In \cref{sec:calibration} we prove that these methods are strongly calibrated in the sense of \citet{Oates2020}, and in \cref{sec:linear_generalisations} we discuss relaxing the stationarity and first degree assumptions.

\subsection{Linear and Stationary Probabilistic Iterative Methods} \label{sec:linear_probabilistic}

For a linear stationary iterative methods of first degree, $P_m(\mathbf{x}_0,\dots,\mathbf{x}_{m-1}) = P(\mathbf{x}_{m-1})$, as described in \cref{ex:richardson},  where
\begin{equation} \label{eq:iterative_method}
	P(\xdummyvec) = \mat{G} \xdummyvec + \mathbf{f}
\end{equation}
for some $\mat{G} \neq \mat{0} \in \reals^{d \times d}$ and $\mathbf{f} \in \reals^d$.
It follows that $P^m(\mathbf{x}_0) = \mat{G}^m \mathbf{x}_0 + \sum_{i=0}^{m-1} \mat{G}^m \mathbf{f}$, where $\mat{G}^0 = \mat{I}$.
We now recall a classical result for linear stationary iterative methods of first degree which will later be useful.
The following is based on \citet[Section 3.3.2 and 3.3.5]{Young1971} and \citet[Section 4.2]{Saad2003}.

\begin{proposition} \label{thm:stationary_iterative_convergence}
	Let $\mat{A}$ be nonsingular, suppose that $\mat{G} \in \reals^{d\times d}$ is such that $\norm{\mat{G}} < 1$ and 
	\begin{align*}
		\mathbf{f} = (\mat{I}_d - \mat{G}) \mat{A}^{-1}\mathbf{b} = (\mat{I}_d - \mat{G})\xtrue  . \numberthis \label{eq:fixed_point_eq}
	\end{align*}
	Then the iterative method
	\begin{equation*}
		\mathbf{x}_{m+1} = \mat{G} \mathbf{x}_m + \mathbf{f} \qquad m\geq 1
	\end{equation*}
	converges to $\xtrue $ for all $\xtrue  \in \reals^d$.
	Furthermore the error in $\mathbf{x}_m$ is controlled by the bound
	\begin{equation*}
		\norm{\xtrue  - \mathbf{x}_m} \leq \norm{\mat{G}}^m \norm{\xtrue  - \mathbf{x}_0}.
	\end{equation*}
\end{proposition}

Now we consider lifting linear stationary iterative methods of first degree into $\mathcal{P}(\mathbb{R}^d)$.
Our main observation is that if $\mu_0$ is Gaussian, then the distribution $\mu_m$ can be computed in closed form using standard formulae for linear transforms of Gaussian distributions.

\begin{proposition} \label{thm:gaussian_stationary_iterative}
	Let $\itermethod$ be a linear, stationary, first degree iterative method.
	Let $\mu_0 = \mathcal{N}(\mathbf{x}_0, \mat{\Sigma}_0)$.
	Then $\mu_m = \mathcal{N}(\mathbf{x}_m, \mat{\Sigma}_m)$, where $\mathbf{x}_m = P^m(\mathbf{x}_0)$ coincides with the iterate from the underlying method, and
$\mat{\Sigma}_m = \mat{G}^m \mat{\Sigma}_0 (\mat{G}^m)^\top$.
	Furthermore we have the following bounds:
	\begin{align*}
		\norm{\xtrue  - \mathbf{x}_m} &\leq \norm{\mat{G}}^m \norm{\xtrue  - \mathbf{x}_0} \qquad
		\norm{\mat{\Sigma}_m} &\leq \norm{\mat{G}}^m \norm{\mat{G}^\top}^m \norm{\mat{\Sigma}_0} .
	\end{align*}
\end{proposition}
\begin{proof}
From elementary properties of Gaussian distributions \citep[Theorem 3.3.3]{Tong1990} we have that $\mu_1 = \mathcal{N}(\mathbf{x}_1, \mat{\Sigma}_1)$ where $\mathbf{x}_1 = \mat{G} \mathbf{x}_0 + \mathbf{f}$ and $\mat{\Sigma}_1 = \mat{G}\mat{\Sigma}_0 \mat{G}^\top$.
This can be continued inductively to achieve the form stated in the proposition for all $m\geq 1$.
The bound on $\norm{\xtrue  - \mathbf{x}_m}$ is a consequence of $\mathbf{x}_m$ coinciding with the classical iterate and \cref{thm:stationary_iterative_convergence}.
The bound on $\mat{\Sigma}_m$ is direct by applying submultiplicativity of the norm $\norm{\quark}$ to $\norm{\mat{G}^m \mat{\Sigma}_0 (\mat{G}^m)^\top}$.
\end{proof}

\begin{remark}
The bound on $\mat{\Sigma}_m$ in \cref{thm:gaussian_stationary_iterative} does not require that $\norm{\quark}$ be the induced norm, only that it is submultiplicative.
As a result, this applies to other matrix norms such as the Frobenius norm, which is submultiplicative but not induced.
\end{remark}

\subsection{Evaluation of Uncertainty Quantification} \label{sec:calibration}

The crucial point that must be addressed in order for probabilistic iterative methods to be useful is whether the covariance matrix $\mat{\Sigma}_m$ relates meaningfully to the error $\mathbf{e}_m = \xtrue  - \mathbf{x}_m$.
It is not possible to provide a satisfactory answer to this question by considering just one linear system; this would be akin to asking whether the number $3$ is meaningfully related to the distribution $\mathcal{N}(0,1)$.
Therefore a collection of linear systems is required so that average-case properties can be discussed.

In \citet{Oates2020} a criterion for meaningful UQ was introduced, building on earlier work such as \citet{dawid1982well,monahan1992proper}.
That work implies the calibration of the PNM can be assessed using an ensemble of linear systems obtained by replacing the right hand side $\mathbf{b}$ with realisations of a random vector $\bm{B} = \mat{A} \bm{X}$, $\bm{X} \sim \mu_0$.
The PNM is then said to be \textit{strongly calibrated} if the true solution $\bm{X}$ is statistically ``plausible'' as a sample from $\mathcal{N}(\mathbf{x}_m,\mat{\Sigma}_m)$, on average with respect to $\bm{X}$, a notion that will be formalised in \Cref{def:strong_pim_nonsingular}.
Note that when $\bm{X}$ is randomised in this way both the mean $\mathbf{x}_m$ and covariance $\mat{\Sigma}_m$ of $\mu_m$ will themselves be random in general\footnote{
	A possible exception occurs if $\itermethod$ is a linear, stationary, first-degree iterative method, when $\mat{\Sigma}_m$ depends only on $\mat{G}$, and for such methods $\mat{G}$ is often independent of $\mathbf{b}$.
	In this case $\mat{\Sigma}_m$ is not random when $\bm{X}$ is randomised.
}, as a consequence of the fact that $\itermethod_\# = \itermethod_\#(\mat{A}, \bm{B})$.
A strongly calibrated PNM provides meaningful UQ, since its output provides a probabilistic representation of uncertainty whose credible sets have correct coverage with respect to realisations of $\bm{X}$.

In this section we will show that linear, stationary, first-degree probabilistic iterative methods are strongly calibrated when a Gaussian $\mu_0$ is used.
This is in contrast to earlier work, where empirical studies in \citet{Cockayne2019b} found that the PNM proposed in that work (called \emph{BayesCG}) failed to be calibrated, though we note that \citet{Reid2020} proposed a particular prior under which BayesCG is calibrated.
Initially we assume that $\mat{\Sigma}_m$ is nonsingular, which implies that $\mat{G}$ must also be nonsingular.

\begin{definition}[Strong calibration, nonsingular case] \label{def:strong_pim_nonsingular}
Fix $\mu_0 \in \mathcal{P}(\mathbb{R}^d)$.
Suppose that a PNM for the solution of \cref{eq:the_system} produces output of the form $\mu_m = \mathcal{N}(\mathbf{x}_m, \mat{\Sigma}_m)$ where $\mat{\Sigma}_m$ is a symmetric positive-definite matrix.
Then the PNM is said to be \emph{strongly calibrated} for $(\mu_0,\mat{A})$ if, when applied to solve a random linear system defined by $\mat{A}$ and $\bm{B} = \mat{A} \bm{X}$, $\bm{X} \sim \mu_0$, it holds for all $m > 0$ that 
\begin{equation} \label{eq:strong_calib_nonsingular}
	\mat{\Sigma}_m^{-\frac{1}{2}}  (\bm{X} - \mathbf{x}_m) \sim \mathcal{N}(\mathbf{0}, \mat{I}_d).
\end{equation}
\end{definition}

\noindent Similar notions of calibration have recently been exploited for verifying the correctness of algorithms for Bayesian computation in \cite{Cook2006,Talts2018}; see \citet{Oates2020} for detail.
Similar ideas have also been explored in the literature on PNM, such as in \citet{Cockayne2019b,Bartels2019,Reid2020}.
Those works explored calibration through a statistic referred to as the \emph{Z-statistic}.
\cref{def:strong_pim_nonsingular} is strictly more general than the Z-statistic, which is obtained by simply taking the norm of \cref{eq:strong_calib_nonsingular}.


The next proposition proves that when $\mat{G}$ is nonsingular, probabilistic iterative methods are strongly calibrated.


\begin{proposition}\label{thm:uq_g_nonsingular}
	Let the assumptions of \cref{def:strong_pim_nonsingular} hold, with $\mu_0 = \mathcal{N}(\mathbf{x}_0,\mat{\Sigma}_0)$ and $\mat{\Sigma}_0$ a positive definite matrix.
	Additionally assume that $\itermethod$ is a linear first degree stationary iterative method with nonsingular $\mat{G}$, and that \cref{eq:fixed_point_eq} holds with probability one when $\itermethod$ is applied to solve a system defined by the right hand side $\bm{B} = \mat{A} \bm{X}$, $\bm{X} \sim \mu_0$.
	Then $\itermethod_\#$ is strongly calibrated for $(\mu_0, \mat{A})$.
\end{proposition}

\begin{proof}
First we consider $\mat{\Sigma}_m^{-1/2}(\xtrue  - \mathbf{x}_m)$ for a \emph{fixed} true solution $\xtrue$; we will complete the proof by randomising $\xtrue $ to obtain the result.
Note that $\mat{\Sigma}_m$ is nonsingular since $\mat{G}$ and $\mat{\Sigma}_0$ are nonsingular.
Now, for each fixed $\xtrue $ and all $m > 0$ we have:
\begin{align*}
	\mat{\Sigma}_m^{-\frac{1}{2}}(\xtrue  - \mathbf{x}_m) 
	&= \mat{\Sigma}_{m-1}^{-\frac{1}{2}} (\mat{G}^{-1} (\xtrue - \mat{G}\mathbf{x}_{m-1} - \mathbf{f}) ) \\
	&= \mat{\Sigma}_{m-1}^{-\frac{1}{2}} (\mat{G}^{-1} (\xtrue - \mathbf{f}) - \mathbf{x}_{m-1}).
\end{align*}
Now we have $\mat{G}^{-1}(\xtrue  - \mathbf{f}) = \xtrue $, from nonsingularity of $\mat{G}$ and \cref{eq:fixed_point_eq}.
It follows inductively over $m$ that
\begin{align*}
	\mat{\Sigma}_m^{-\frac{1}{2}}(\xtrue  - \mathbf{x}_m)&= \mat{\Sigma}_{m-1}^{-\frac{1}{2}}(\xtrue  - \mathbf{x}_{m-1}) \\
		&= \mat{\Sigma}_0^{-\frac{1}{2}}(\xtrue  - \mathbf{x}_0).
\end{align*}
Thus, if we now randomise $\xtrue $ according to $\bm{X} \sim \mu_0 = \mathcal{N}(\mathbf{x}_0, \mat{\Sigma}_0)$, we obtain $\mat{\Sigma}_m^{-1/2}(\bm{X} - \mathbf{x}_m ) \sim \mathcal{N}(\mathbf{0}, \mat{I}_d)$, completing the proof.
\end{proof}

\begin{remark}
	The only demand \cref{thm:uq_g_nonsingular} makes of $\itermethod$ is that \cref{eq:fixed_point_eq} is almost surely satisfied; it does not require that $\norm{\mat{G}} < 1$.
	Thus strong calibration of a PNM does not imply that $\mu_m$ contracts to the truth, only that $\mu_m$ should be a fair reflection of the size of the error.
	For example, if $\itermethod$ diverges for some $\mathbf{x}_0$ it is natural that $\mu_m$ should tend to a distribution with infinite variance as $m\to\infty$.
\end{remark}

The assumption of nonsingular $\mat{G}$ permits a straightforward proof for \Cref{thm:uq_g_nonsingular}, but unfortunately $\mat{G}$ may be singular even for such elementary methods as the Jacobi iterations.
The next definition adapts \cref{def:strong_pim_nonsingular} to the case where $\mat{G}$, and therefore also $\mat{\Sigma}_m$, are singular.
It simplifies the subsequent presentation to focus on the case where $\mat{\Sigma}_m$ does not depend on $\mat{b}$.
To the best of our knowledge this is the case for the majority of stationary iterative methods. 

\begin{definition}[Strongly calibrated, singular case] \label{def:strong_pim_singular}
Fix $\mu_0 \in \mathcal{P}(\mathbb{R}^d)$.
Suppose that a PNM for the solution of \cref{eq:the_system} produces output of the form $\mu_m = \mathcal{N}(\mathbf{x}_m, \mat{\Sigma}_m)$ where $\mat{\Sigma}_m$ is a positive semidefinite matrix with rank $0 < r < d$, with $\mat{\Sigma}_m$ not depending on the right hand side $\mathbf{b}$. 
Let $\mat{R} \in \mathbb{\reals}^{d \times r}, \mat{N} \in \mathbb{\reals}^{d \times (d-r)}$ be matrices such that $\textup{range}(\mat{R}) = \textup{range}(\mat{\Sigma}_m)$ and $\textup{range}(\mat{N}) = \textup{ker}(\mat{\Sigma}_m)$.
Then the PNM is said to be \emph{strongly calibrated} for $(\mu_0, \mat{A})$ if, when applied to solve a random linear system defined by $\mat{A}$ and $\bm{B} = \mat{A} \bm{X}$, $\bm{X} \sim \mu_0$, the following two conditions are satisfied:
\begin{enumerate}
	\item $(\mat{R}^\top\mat{\Sigma}_m\mat{R})^{-\frac{1}{2}} \mat{R}^\top (\bm{X} - \mathbf{x}_m) \sim \mathcal{N}(\mathbf{0}, \mat{I}_r)$.
	\item $\mat{N}^\top (\bm{X} - \mathbf{x}_m) = \mathbf{0}$.
\end{enumerate}
\end{definition}

\noindent This definition is an intuitive extension of \cref{def:strong_pim_nonsingular} to the case of singular $\mat{\Sigma}_m$; it demands that in any subspace of $\reals^d$ in which $\mat{\Sigma}_m$ is nonzero, the PNM is strongly calibrated as in \cref{def:strong_pim_nonsingular}, and in any subspace in which it is zero and thus no uncertainty remains, $\mathbf{x}_m$ is identically equal to the true solution $\bm{X}$.


We then have the following result, the proof of which is provided in \cref{supp:sec:proof_calibration_diagonalisable}.
The intuition behind the proof in the singular case is the same as in the nonsingular case, but  additional technical effort is required to project into the null space of $\mat{\Sigma}_m$.

\begin{proposition}\label{thm:calibration_diagonalisable}
	Let $\mu_0 = \mathcal{N}(\mathbf{x}_0,\mat{\Sigma}_0)$ where $\mat{\Sigma}_0$ is a positive definite matrix.
	Let $\itermethod$ be a linear first degree stationary iterative method such that \cref{eq:fixed_point_eq} holds with probability one when $\itermethod$ is applied to solve a system defined by the right hand side $\bm{B} = \mat{A} \bm{X}$, $\bm{X} \sim \mu_0$.
	Suppose that $\mat{G}$ is independent of $\bm{B}$, and that $\mat{G}$ is diagonalisable with rank $0 < r \leq d$.
	Then the probabilistic iterative method $\itermethod_\#$ is strongly calibrated for $(\mu_0, \mat{A})$.
\end{proposition}


\begin{remark}
	Since in \Cref{def:strong_pim_singular} the matrix $\mat{\Sigma}_m$ does not depend on $\bm{B}$, both $\mat{R}$ and $\mat{N}$ can be fixed matrices independent of $\bm{X}$.
	Furthermore while the columns of $\mat{R}$ and $\mat{G}$ must be bases of the range and kernel of $\mat{\Sigma}_m$ respectively, \cref{def:strong_pim_singular} and \cref{thm:calibration_diagonalisable} are basis-independent.
\end{remark}

\Cref{thm:uq_g_nonsingular,thm:calibration_diagonalisable} provide a clear and defensible sense in which the output $\mu_m$ from a probabilistic iterative method $\itermethod_\#$, arising from a linear first degree stationary iterative method $\itermethod$, can be considered to be meaningful.
Specifically, one has a guarantee that the unknown solution is indistinguishable, in a statistical sense, from samples drawn from $\mu_m$.
Thus one may interpret $\mu_m$ as quantifying uncertainty with respect to the unknown true value of $\xtrue $ in \cref{eq:the_system}.


\subsection{Generalisations} \label{sec:linear_generalisations}

Here we discuss generalisations to both non-stationary and higher degree iterative methods, while remaining in the linear framework.

\subsubsection{Non-Stationary Methods}

In a non-stationary linear iterative method of first degree \cite[Chapter 9]{Young1971}, the iteration is of the form:
\begin{equation} \label{eq:non_stationary}
	\mathbf{x}_m = \mat{G}_m \mathbf{x}_{m-1} + \mathbf{f}_m
\end{equation}
where $\mathbf{f}_m \in \reals^d$ and $\mat{G}_m \in \reals^{d \times d}$ for all $m \geq 0$.
The map $P^m$ is then of the form:
\begin{align*}
	P^m (\mathbf{x}_0) &= \hat{\mat{G}}_m\mathbf{x}_0 + \hat{\mathbf{f}}_m \\
	\hat{\mat{G}}_m &= \prod_{i=1}^m \mat{G}_m \qquad
	\hat{\mathbf{f}}_m =\mathbf{f}_m + \sum_{i=1}^m \left( \prod_{j=i+1}^m \mat{G}_j \right) \mathbf{f}_i .
\end{align*}
From this it follows by an identical argument to \cref{thm:gaussian_stationary_iterative} that $\mu_m = \mathcal{N}(\mathbf{x}_m, \mat{\Sigma}_m)$, with $\mathbf{x}_m = \hat{\mat{G}}_m\mathbf{x}_0 + \hat{\mathbf{f}}_m$ and $\mat{\Sigma}_m = \hat{\mat{G}}_m\mat{\Sigma}_0 \hat{\mat{G}}_m$.

Considering the consistency of the implied probabilistic iterative method, as in the stationary setting, $\mathbf{x}_m$ coincides with the classical iterate.
Furthermore, \cite{Young1971} notes that the iteration from \cref{eq:non_stationary} converges to $\xtrue$ only if $\hat{\mat{G}}_m \to \mat{0}$.
In this event clearly $\mat{\Sigma}_m \to \mat{0}$, and so provided the underlying iterative method converges, $\mu_m$ converges to an atomic mass on $\xtrue $ as $m \to \infty$.

From the perspective of calibration of UQ, the proofs in \cref{sec:calibration} do not apply to non-stationary iterative methods $\itermethod$ since those proofs exploit that $\mat{\Sigma}_m = \mat{G}^m \mat{\Sigma}_0 (\mat{G}^m)^\top$, which no longer holds in the non-stationary setting.
However if one instead directly assumes $\hat{\mat{G}}_m$ to be diagonalisable for each $m$, the proof of \cref{thm:uq_g_nonsingular} would need only minor modifications to establish that the associated probabilistic iterative method $\itermethod_\#$ is strongly calibrated in the non-stationary setting.

\subsubsection{Higher Degree Methods}

Modifying \cref{def:iterative_method} to allow methods of higher degree requires changing the space on which $\mu$ is defined, and the domain of $P_m$ (and by extension $(P_m)_\#$), to a Cartesian product of $s$ instances of $\reals^d$.

In terms of such methods, when $s=2$ \cite[Chapter 16]{Young1971} the iteration takes the form
\begin{equation} \label{eq:second_degree}
	\mathbf{x}_m = \mat{G} \mathbf{x}_{m-1} + \mat{H}\mathbf{x}_{m-2} + \mathbf{k}
\end{equation}
where $\mat{G}, \mat{H} \in \reals^{d \times d}$ and $\mathbf{k} \in \reals^d$.
While second degree methods are seldom used in practice, higher order methods can accelerate convergence and raise some interesting statistical questions.
These methods are analysed by augmenting the space as follows, to obtain a first degree linear stationary iterative method on $\reals^{2d}$:
\begin{equation*}
	\begin{pmatrix}
		\mathbf{x}_{m-1} \\
		\mathbf{x}_m
	\end{pmatrix}
	=
	\begin{pmatrix}
		\mat{0} & \mat{I}_d \\ \mat{H} & \mat{G}
	\end{pmatrix}
	\begin{pmatrix}
		\mathbf{x}_{m-1} \\
		\mathbf{x}_m
	\end{pmatrix}
	+
	\begin{pmatrix}
		\mathbf{0} \\ \mathbf{k}
	\end{pmatrix} =
	\tilde{\mat{G}} 
	\begin{pmatrix}
		\mathbf{x}_{m-1} \\
		\mathbf{x}_m
	\end{pmatrix}
	+
	\tilde{\mathbf{k}} .
\end{equation*}
Convergence of the iterate, and hence the covariance in \cref{thm:gaussian_stationary_iterative}, then requires $\rho(\tilde{\mat{G}}) < 1$.
Similarly, provided $\tilde{\mat{G}}$ satisfies the assumptions in \cref{sec:calibration}, $\mu_m$ will provide meaningful UQ according to \cref{def:strong_pim_nonsingular,def:strong_pim_singular}.

An interesting technicality for higher degree methods is that, whereas in first degree methods only an initial iterate $\mathbf{x}_0$ must be supplied, in second degree methods both the iterates $\mathbf{x}_0$ and $\mathbf{x}_1$ are required.
This raises a challenge in the probabilistic framework because it is not clear how one should specify an initial distribution jointly over $\mathbf{x}_0$ and $\mathbf{x}_1$.
While expert knowledge may be exploited to build a distribution over $\mathbf{x}_0$, the same is not true of $\mathbf{x}_1$.
Several possible approaches are considered experimentally in \cref{sec:experiments}.

\section{Beyond Linearity} \label{sec:extensions}

In the non-Gaussian and non-linear setting it is significantly more difficult to formulate an appropriate sense in which a PNM can be considered to be strongly calibrated.
Instead, in this section we adopt a strictly weaker notion called \emph{weak calibration}, which is simply defined and can be empirically tested.
In \cref{sec:weak_calib} we present that definition and in \cref{sec:testing_weak} discuss statistical tests for weak calibration which will be applied in \cref{sec:experiments} when nonlinear iterative methods are assessed.

\subsection{Weakly Calibrated Probabilistic Iterative Methods} \label{sec:weak_calib}

The chief issue with \cref{def:strong_pim_nonsingular,def:strong_pim_singular} is that in order to define strong calibration we require that $\mu_m$ is Gaussian.
This is problematic because Gaussian distributions are unable to express all initial beliefs about components of $\xtrue $, and because the linear iterative methods which result in a Gaussian $\mu_m$ are less widely-used compared to nonlinear iterative methods, such as CG.
Therefore we turn to an alternative, weaker sense in which the output $\mu_m$ from a (possibly nonlinear) probabilistic iterative method can be considered to be meaningful.

Our notion of weak calibration is also due to \citet{Oates2020}, and will now be defined.
In the same setting as \Cref{sec:calibration}, we fix $\mat{A}$ and randomly generate a right hand side $\bm{B} = \mat{A} \bm{X}$, $\bm{X} \sim \mu_0$.
Then, conditional on $\bm{X}$ and for each $m > 0$, we introduce a second random variable $\bm{Y}^{(m)} | \bm{X} \sim \mu_m$ that is sampled from the output $\mu_m$ of the PNM applied to solve the linear system defined by $\mat{A}$ and $\bm{B}$.
Let $\bm{Y}^{(m)}$ denote the random variable obtained by marginalising $\bm{Y}^{(m)}|\bm{X}$ over realisations of $\bm{X}$.


\begin{definition}[Weakly calibrated] \label{def:weak_pim}
Fix $\mu_0 \in \mathcal{P}(\mathbb{R}^d)$.
A PNM for the solution of \cref{eq:the_system} is said to be \emph{weakly calibrated} to $(\mu_0, \mat{A})$ if, when applied to solve a random linear system defined by $\mat{A}$ and $\bm{B} = \mat{A}\bm{X}$, $\bm{X} \sim \mu_0$, and when $\bm{Y}^{(m)} | \bm{X} \sim \mu_m$, it holds for all $m > 0$ that $\bm{Y}^{(m)}$ has marginal distribution
	\begin{equation}
		\bm{Y}^{(m)} \sim \mu_0 . \label{eq: weak calib}
	\end{equation}
\end{definition}

\noindent \cref{eq: weak calib} is sometimes called the \emph{self-consistency property} and, as with strong calibration, the notion of weak calibration has previously been exploited to verify the correctness of algorithms for Bayesian computation \citep{geweke2004getting}.
\citet[Lemma 2.19]{Oates2020} establishes that strong calibration implies weak calibration.
Although weaker than strong calibration, \cref{def:weak_pim} allows for statistical tests of distributional equality to be used to assess the quality of the uncertainty quantification provided by a PNM whose output is non-Gaussian.

\begin{remark}[Strong versus weak calibration]
From a simulation perspective, we can intuitively think about strong and weak calibration in the following terms:
\begin{enumerate}
	\item draw $\bm{X} \sim \mu_0$,
	\item compute output $\mu_m$ from the probabilistic iterative method $\itermethod_\#(\mat{A}, \mat{A}\bm{X} )$,
	\item draw $\bm{X}' \sim \mu_m$,
\end{enumerate}
then, in strong calibration we
\begin{enumerate}
	\item[4.] compare $\bm{X}$ to $\bm{X}'$.
\end{enumerate}
while in weak calibration we
\begin{enumerate}
	\item[4.] independently draw $\bm{X}'' \sim \mu_0$ and compare $\bm{X}''$ to $\bm{X}'$.
\end{enumerate} 
Thus in strong calibration a conditional comparison is performed, while in weak calibration only a marginal comparison is performed.
\end{remark}

\subsection{Testing for Weak Calibration} \label{sec:testing_weak}

We now present a statistical test to determine whether a PNM is weakly calibrated.
For convenience we let $\nu_m$ denote the distribution of $\bm{Y}^{(m)}$, so that we aim to test whether $\nu_m = \mu_0$.
Since $\nu_m$ does not necessarily have a closed form but it is possible to access samples from $\nu_m$, we aim to perform a goodness-of-fit test to determine whether such samples are consistent with being drawn from $\mu_0$.
In this work we adopt a general purpose goodness-of-fit test based on \emph{maximum mean discrepancy} (MMD), due to \citet{Gretton2012}, which we briefly describe next.


\begin{definition}[Maximum mean discrepancy] \label{def: MMD}
	Let $\mu, \nu \in \mathcal{P}(\reals^d)$ and let $\mathcal{F}$ be a set of real-valued, $\mu$ and $\nu$-integrable functions on $\reals^d$.
	Then the MMD between $\mu$ and $\nu$, based on $\mathcal{F}$, is given by
	\begin{equation*}	
		\textsc{mmd}_\mathcal{F}(\mu, \nu) := \sup_{f \in \mathcal{F}} \left(\int f(\xdummyvec) \;\mu(\wrt \xdummyvec) - \int f(\xdummyvec)\; \nu(\wrt \xdummyvec) \;\right) .
	\end{equation*}
\end{definition}

\noindent \citet{Gretton2012} considered taking $\mathcal{F}$ to be a unit ball in a reproducing kernel Hilbert space (RKHS), showing that when the RKHS is chosen judiciously then MMD is a metric on $\mathcal{P}(\reals^d)$.
Moreover, this choice ensures that an unbiased estimator for MMD can be constructed, as will now be explained.
Recall that an RKHS is associated with a symmetric positive definite \emph{kernel} $k : \reals^d \times \reals^d \to \reals$; we emphasise this using the notation $\mathcal{F} \equiv \mathcal{F}_k = \{f \in \mathcal{H}_k : \|f\|_{\mathcal{H}_k} \leq 1 \}$ where $\mathcal{H}_k$ is the unique RKHS with kernel $k$ and $\|\cdot\|_{\mathcal{H}_k}$ is the norm in $\mathcal{H}_k$.
Define the \emph{kernel mean embedding} of $\mu$ in $\mathcal{H}_k$ as $\mu[k]$ where $\mu[k](\xdummyvec) := \int k(\xdummyvec, \xdummyvec') \mu(\wrt \xdummyvec')$.
Then \citet[Lemma 4]{Gretton2012} asserts that $\textsc{MMD}_{\mathcal{F}_k}(\mu, \nu)$ can be expressed as a difference between the kernel mean embeddings of $\mu$ and $\nu$:
\begin{equation} \label{eq:mmd_rkhs}
	\textsc{mmd}_{\mathcal{F}_k}(\mu, \nu) := \norm{\mu[k] - \nu[k]}_{\mathcal{H}_k} .
\end{equation}
For convenient choices of $k$ and $\mu$ it may be possible to compute $\mu[k]$ in closed-form, but in general one must resort to approximating \cref{eq:mmd_rkhs} based on samples from one or both of $\mu$ and $\nu$.
Given independent samples $\bm{X}_1,\dots,\bm{X}_N \sim \mu_0$ and $\bm{Y}_1^{(m)}, \dots, \bm{Y}_N^{(m)} \sim \nu_m$, we define an estimator
\begin{align} \label{eq:mmd_statistic}
	\hspace{-5pt}
	\widehat{\textsc{mmd}}_{\mathcal{F}_k}^2 &:= \frac{1}{N(N-1)} \sum_{\substack{i,j = 1 \\ i \neq j}}^N  k(\bm{X}_i, \bm{X}_j) + k(\bm{Y}_i^{(m)}, \bm{Y}_j^{(m)}) - k(\bm{X}_i, \bm{Y}_i^{(m)}) - k(\bm{Y}_i^{(m)}, \bm{X}_j) ,
\end{align}
which can be verified to be an unbiased estimator of $\textsc{mmd}_{\mathcal{F}_k}(\mu, \nu_m)^2$ provided that, in addition to having the stated distribution, the samples $\bm{Y}_1^{(m)},\dots,\bm{Y}_N^{(m)}$ are generated independently from the samples $\bm{X}_1,\dots,\bm{X}_N$.

The statistic in \cref{eq:mmd_statistic} enables a goodness-of-fit test to be performed, and the distribution of this test statistic under the null hypothesis $\nu_m = \mu_0$ may be estimated using a standard bootstrap procedure as described in \citet[Section 5]{Gretton2012}.
Having obtained $M$ approximate samples from the distribution of \cref{eq:mmd_statistic} using the bootstrap, we determine a threshold for a prescribed power level $\alpha \in (0,1)$ by computing a $(1-\alpha)$-quantile of this empirical distribution.
This procedure will be used in \Cref{sec:experiments}, next, to empirically test whether PNM are weakly calibrated.

\section{Empirical Assessment} \label{sec:experiments}
The aim of this section is to empirically assess our proposed probabilistic iterative methods.
For this purpose we consider the problem of inverting a linear system that arises when building a kernel interpolant.
Our aim is not to address the problem of computing kernel interpolants \textit{per se}, as many powerful methods exist for this task, but this problem serves as a convenient test-bed in which probabilistic iterative methods can be examined.

\subsection{Problem Definition}

Consider a dataset consisting of pairs $(z_i, y_i)$, $i=1,\dots, d$, $d\in\mathbb{N}$, where the $z_i \in [0,1]$ are distinct locations at which observations $y_i \in \mathbb{R}$ of some physical phenomenon were obtained. 
The aim is to compute a interpolant of this dataset, that is, a function $g : [0,1] \to \reals$ which is such that $g(z_i) = y_i$ for all $i =1,\dots,d$.
For a given symmetric positive definite kernel $c : [0,1] \times [0,1] \rightarrow \mathbb{R}$, we consider an interpolant of the form
\begin{equation}
g(z) := \sum_{i=1}^d x_i c(z,z_i)
\label{eq:gs}
\end{equation}
and note that there is a unique set of weights $x_i \in \mathbb{R}$ such that the interpolation equations
\begin{equation*}
g(z_i) = y_i, \qquad i = 1,\dots,d
\end{equation*}
are satisfied.
The vector $\xtrue = (x_1,\dots,x_d)^\top$ of such weights satisfies the $d$-dimensional linear system in \cref{eq:the_system} with $A_{i,j} = c(z_i,z_j)$ and $\mathbf{b} = (y_1,\dots,y_d)^\top$.

This linear system is representative of linear systems that are widely encountered in statistics and machine learning, and naturally a variety of methods have been proposed to circumvent the need to solve them; for example, based on reducing the degrees of freedom of the parametric function $g$ so that the dataset is only approximately interpolated.
Our aim is to use a finite number of iterations, $m$, of a probabilistic iterative method on the full problem in \cref{eq:the_system} and to lift the distribution $\mu_m$ over the unknown solution vector $\xtrue$ into the function space spanned by functions of the form in \cref{eq:gs}.
This enables uncertainty due to limited computation to be interpreted in the domain on which the interpolation problem was defined.

The condition number of $\mat{A}$ depends on the spectrum of the kernel $c$ and the closeness of the elements in $\set{z_1,\dots,z_d}$.
For kernels with rapidly decaying spectrum, such as the squared exponential kernel
\begin{equation} \label{eq:sqexp}
c(x,y) = \exp\left(-\frac{\norm{x-y}_2^2}{2\ell^2} \right)
\end{equation}
with length-scale parameter $\ell > 0$, it is common for $\mat{A}$ to be badly conditioned.
Thus even when $d$ is small, direct solution of \cref{eq:the_system} can be difficult and careful numerical analysis is required.

A dataset of size $d = 440$ was generated, with $(z_i)_{i=1,\dots,d}$ consisting of 20 evenly spaced points in $[0,0.1]$, 400 evenly spaced points in $[0.2,0.8]$ and 20 evenly spaced points in $[0.9, 1]$, and $y_i = f(z_i)$ where $f(z) = 1_{z < 0.5} \sin(2 \pi z) + 1_{z \geq 0.5} \sin( 4 \pi z)$.
The parameter $\ell = 0.0012$ was used, which produces a system for which a direct solver can be used, so that a ground-truth is accessible, but which is not entirely trivial.

\subsection{Choice of \texorpdfstring{$\mu_0$}{Initial Distribution}} \label{sec:prior}

For the initial distribution $\mu_0$ several candidates were considered.
Firstly a \textsc{default} choice given by $\mu_0 = \mathcal{N}(\mathbf{0}, \mat{I}_d)$ which can be interpreted as a lack of \emph{a priori} insight.
Secondly the \textsc{natural} choice $\mu_0 = \mathcal{N}(\mathbf{0}, \mat{A}^{-1})$ which incorporates the structure of $\mat{A}$ into the initial distribution, and has been noted to have desirable theoretical properties in the related work of \citet{Cockayne2019b,Hennig2015b}.
We note that the \textsc{natural} initial distribution is not a practical choice in general as it requires computation of $\mat{A}^{-1}$. 

The third initial distribution we consider is applicable only in settings where a small number of \textit{ansatz} solutions (i.e. guesses) are provided, perhaps obtained by expert knowledge of the system at hand.
Let ${\xtrue}_{i}$, $i=1,\dots,N$, be these \textit{ansatz} solutions; we use these to estimate the scaling parameter $\nu^2$ for an initial distribution $\mu_0 = \mathcal{N}(\mathbf{0}, \nu^2\mat{\Sigma}_0)$ where $\mat{\Sigma}_0$ is fixed.
Maximum likelihood estimation yields the estimator
\begin{equation*}
	\nu^2_{\textup{opt}} := \frac{1}{Nd}\sum_{i=1}^N \norm{{\xtrue}_{i}}_{\mat{\Sigma}_0^{-1}}^2 ,
\end{equation*}
which can be seen to adapt to the scale of the problem at hand; we call this approach \textsc{opt}.
In the experiments below where this approach is used we assume that $\mat{\Sigma}_0 = \mat{I}_d$.
We used $N=5$ \emph{ansatz} solutions, obtained by sampling $5$ right-hand-sides $\mathbf{B}_1,\dots,\mathbf{B}_5 \sim \mathcal{N}(\mathbf{0}, \mathbf{I}_d)$ and computing ${\xtrue}_{i} = \mat{A}^{-1} \mathbf{B}_i$.
We note that this does not result in an entirely fair comparison since $5$ exact solutions to the linear system are used to construct the initial distribution.
One could consider instead using only approximate solutions, but this introduces additional degrees of freedom into the assessment.
Since the focus of this paper goes beyond selecting $\mu_0$, we simply use exact solutions within \textsc{opt} for the assessment.

\subsection{Results in Function Space}

\begin{figure}[t!]
\centering
\begin{subfigure}{\textwidth}
	\centering
	\includegraphics[width=0.8\textwidth]{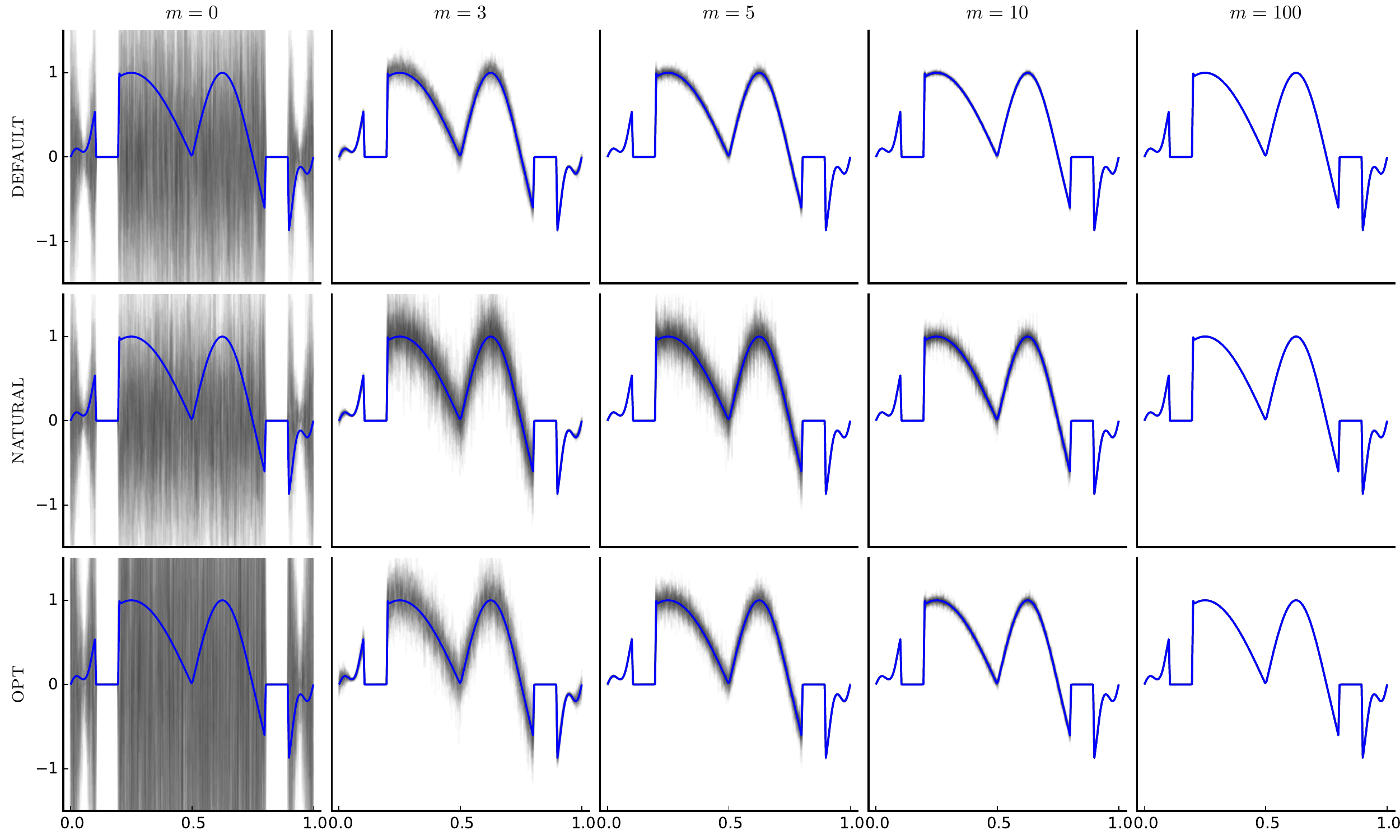}
	\caption{Default Step Size} \label{fig:richardson_default}
\end{subfigure}

\begin{subfigure}{\textwidth}
	\centering
	\includegraphics[width=0.8\textwidth]{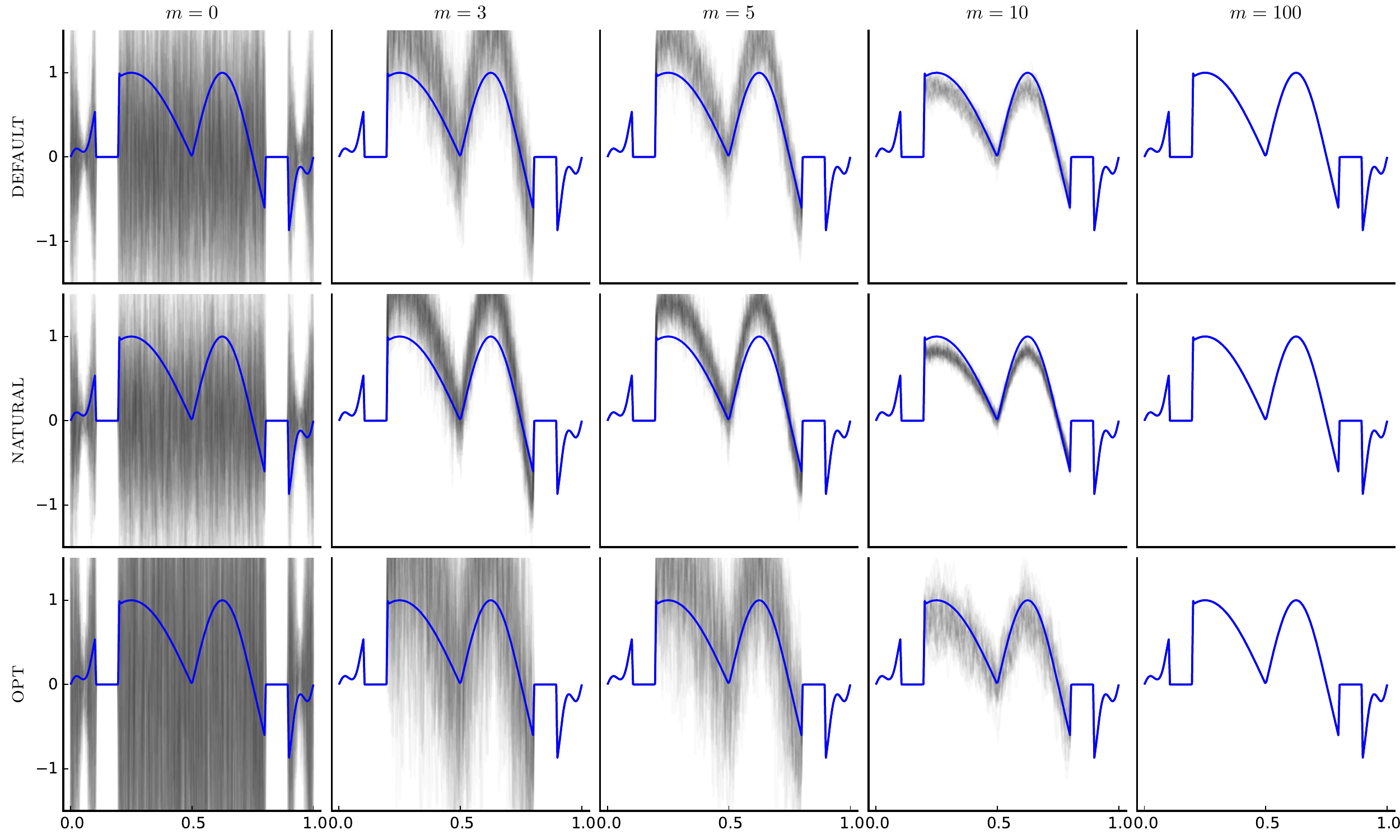}
	\caption{Optimal Step Size} \label{fig:richardson_optimal}
\end{subfigure}
\caption{Samples from the distributional output of a probabilistic iterative method based on Richardson iteration, used to solve an interpolation problem and visualised in the physical domain in which the interpolant is defined.
The rows in each figure represent the three choices of initial distribution described in \cref{sec:prior}.
In each panel we present 50 samples (grey curves) from the output of the probabilistic iterative method after $m$ iterations have been performed.
The interpolant, corresponding to the exact solution of the linear system, is also shown in blue.
}
\label{fig:richardson_stationary}
\end{figure}

In this section we examine the resulting distributions $\mu_m$ from application of a number of probabilistic iterative methods to the problem above, for each choice of initial distribution from \cref{sec:prior}.

\paragraph{Stationary Iterative Methods}
We first consider Richardson's iteration with a constant step size.
Since this method is stationary and linear, the theoretical results obtained in \cref{sec:linear_probabilistic_methods} apply.
The step size $\omega$ was set to either the optimal value, $\omega = 2 / (\lambda_{\text{min}}(\mat{A}) + \lambda_{\text{max}}(\mat{A}))$, that minimises the spectral radius of $\mathbf{G}$, or a default value $\omega = 2/3$.
Jacobi's method was also considered, but in our simulations the results were virtually identical for this problem, so they are not presented.

\cref{fig:richardson_stationary} displays samples (grey curves) from each of the probabilistic iterative methods that we considered and the blue curve represents the exact kernel interpolant.
For each probabilistic iterative method, the output was seen to contract around the exact solution as the number $m$ of iterations is increased.
Interestingly, very little variation is observed in the intervals $[0,0.1]$ and $[0.9,1]$, which accords with the fact that the interpolant is being approximated well in these regions - this suggests that the distributional output can act as a local error indicator.


\begin{figure}
\centering
	\includegraphics[width=0.8\textwidth]{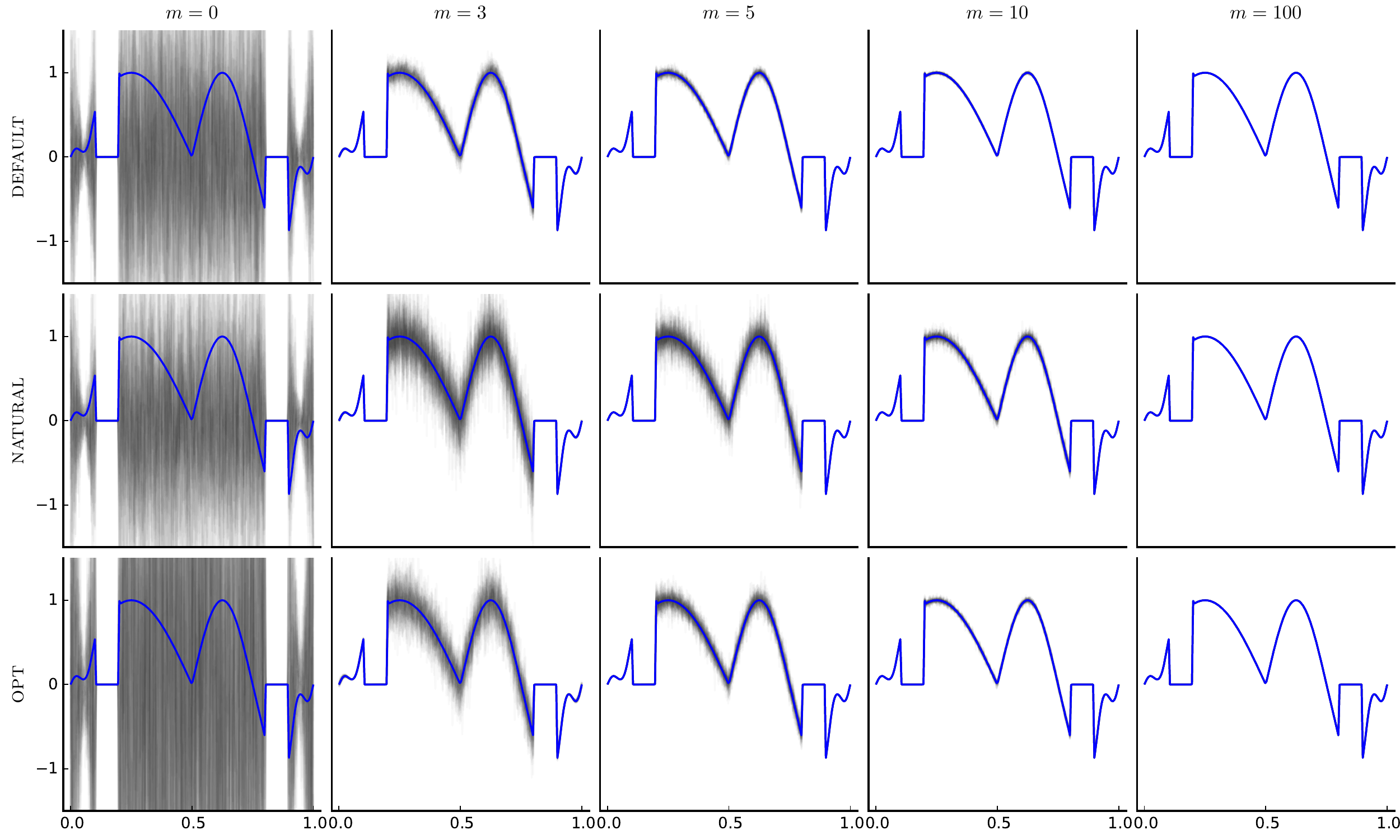}
	\caption{As in \cref{fig:richardson_stationary}, but with step size chosen adaptively.} \label{fig:richardson_adaptive}

	\includegraphics[width=0.8\textwidth]{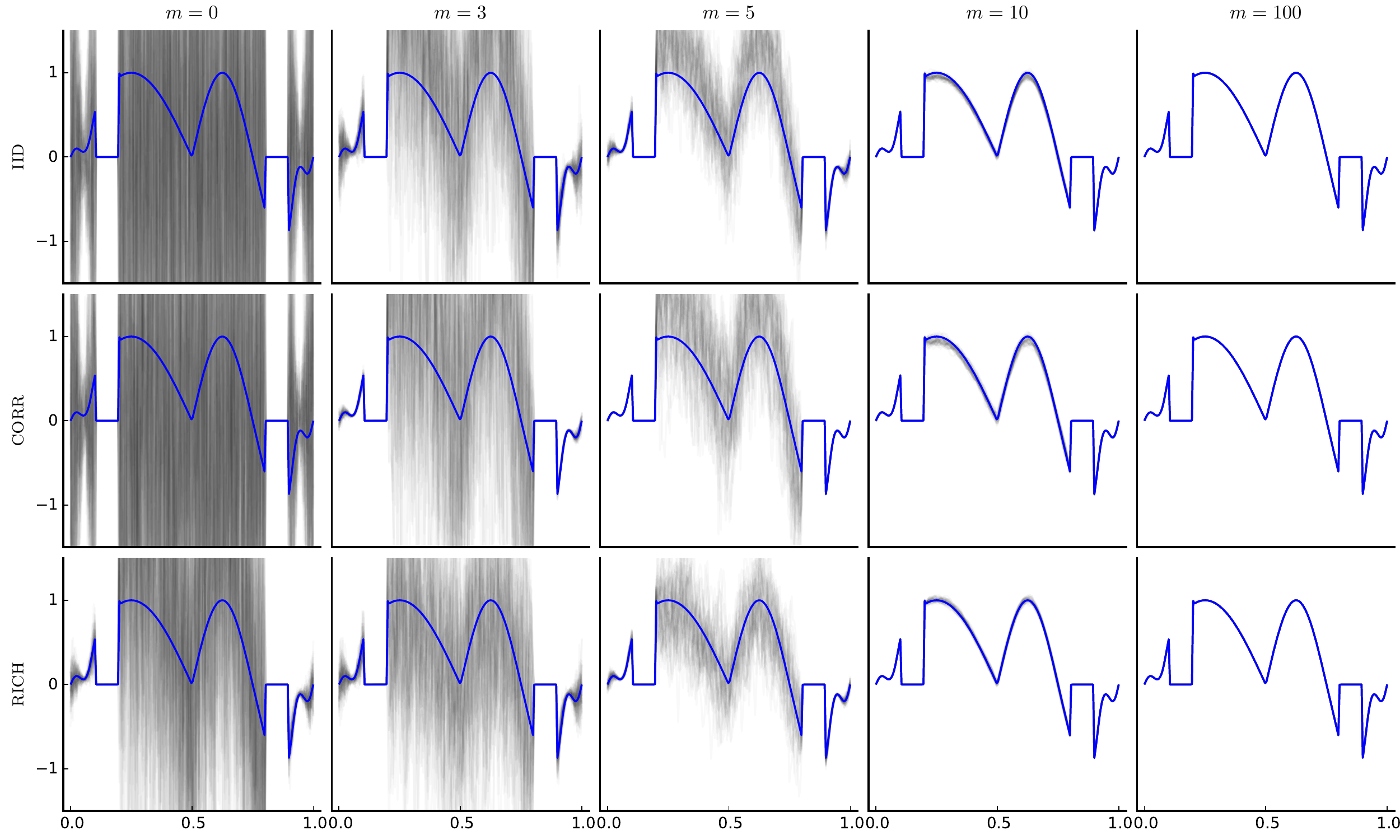}
	\caption{
		A probabilistic iterative method based on a second degree version of Richardson iteration, as described in \cref{sec:experiments}.
		Each row uses \textsc{opt} as the initial distribution for $\mathbf{x}_0$ and a different initial distribution for $\mathbf{x}_1$, as described in the main text.
	} \label{fig:richardson_second_order}
\end{figure}

\paragraph{Non-Stationary and Higher-Order Methods}
We now consider non-stationary and higher-order iterative methods.
As discussed in \cref{sec:linear_generalisations}, these methods are expected to be strongly calibrated as they are still linear, though calibration has not been rigorously established.
For the non-stationary scheme we considered Richardson iteration again but with the step-size chosen adaptively, with $\omega_m = \mathbf{r}_m^\top \mathbf{A} \mathbf{r}_m / \|\mathbf{A} \mathbf{r}_m\|_2^2$ minimising the Euclidean norm of the residual $\mathbf{r}_{m+1} = \mathbf{b} - \mathbf{A} \mathbf{x}_{m+1}$.
Results for the non-stationary scheme are presented in \cref{fig:richardson_adaptive}, with qualitative behaviour appearing to be similar to that with the default step size from \cref{fig:richardson_default}.
Since the non-stationary scheme is better able to adapt to the problem at hand, this seems a more prudent choice than an arbitrary $\omega = 2/3$, though we note that the calibration of this method remains to be assessed empirically; this will be considered in \cref{sec:experiment_calibration}.

As an example of a higher-order iterative method, we consider a second-degree version of Richardson iteration presented in \citet{Young1972}.
In this method the iteration is of the form
\begin{equation*}
	\mathbf{x}_m = \gamma \sigma \left( \frac{2}{\beta - \alpha} \mat{G} - \frac{\beta + \alpha}{\beta - \alpha}\right) \mathbf{x}_{m-1}
		+ (1 - \gamma) \mathbf{x}_{m-2}
		+ \frac{2 \gamma \sigma}{\beta - \alpha}\mathbf{f}
\end{equation*}
where $\mat{G}$ and $\mathbf{f}$ are as given in the classical first-order Richardson iteration from \cref{ex:richardson}, with optimal step size $\omega = 2 / (\lambda_{\text{min}}(\mat{A}) + \lambda_{\text{max}}(\mat{A}))$, while $\alpha=\lambda_{\min}(\mat{G}), \beta=\lambda_{\max}(\mat{G})$ and
\begin{equation*}
	\sigma = \frac{\beta - \alpha}{2 - (\beta - \alpha)} \qquad \gamma = \frac{2}{1 + \sqrt{1 - \sigma^2}}.
\end{equation*}
Recall that for a second degree probabilistic iterative method, a joint initial distribution must be specified for $\mathbf{x}_0$ and $\mathbf{x}_1$.
The distribution assigned to $\mathbf{x}_0$ was fixed to \textsc{opt}, since, in the results for $s=1$ (to follow), this appeared to provide better UQ across different choices of $\omega$.
Three choices were considered for initial distributions for $\mathbf{x}_1$: \textsc{iid}, in which $\mathbf{x}_1$ is an independent copy of $\mathbf{x}_0$, \textsc{corr}, in which $\mathbf{x}_1$ is identical to $\mathbf{x}_0$ and \textsc{rich}, in which $\mathbf{x}_1$ is obtained from $\mathbf{x}_0$ by performing one iteration of Richardson iteration with optimal step size.
Note that both \textsc{iid} and \textsc{corr} yield the same marginal distribution for $\mathbf{x}_1$, but the joint distributions differ.

\cref{fig:richardson_second_order} displays samples from the output of the probabilistic iterative methods just described.
Qualitatively, the results appear to be similar to those from \cref{fig:richardson_optimal} with initial distribution \textsc{opt}, as one would expect given that $\mu_0$ in all three rows is that same distribution.
Of the three choices for $\mu_1$, \textsc{rich} appears to contract marginally faster, though in all three methods the improvement over the first order method from \cref{fig:richardson_optimal} appears to be negligible.

\paragraph{Nonlinear Methods}
Here we consider a probabilistic iterative method based on CG, which is the most widely-used of the iterative methods we consider, but for which our theoretical results on strong calibration do not hold.
Results are displayed in \cref{fig:cg}.
Convergence is clearly seen to be faster than in the other methods considered, though qualitatively the samples obtained otherwise seem to be similar.
This hints at the results from the next section, in which we will see that CG is weakly calibrated for this problem and for the initial distributions that we considered.

\begin{figure}
	\centering
	\includegraphics[width=0.8\textwidth]{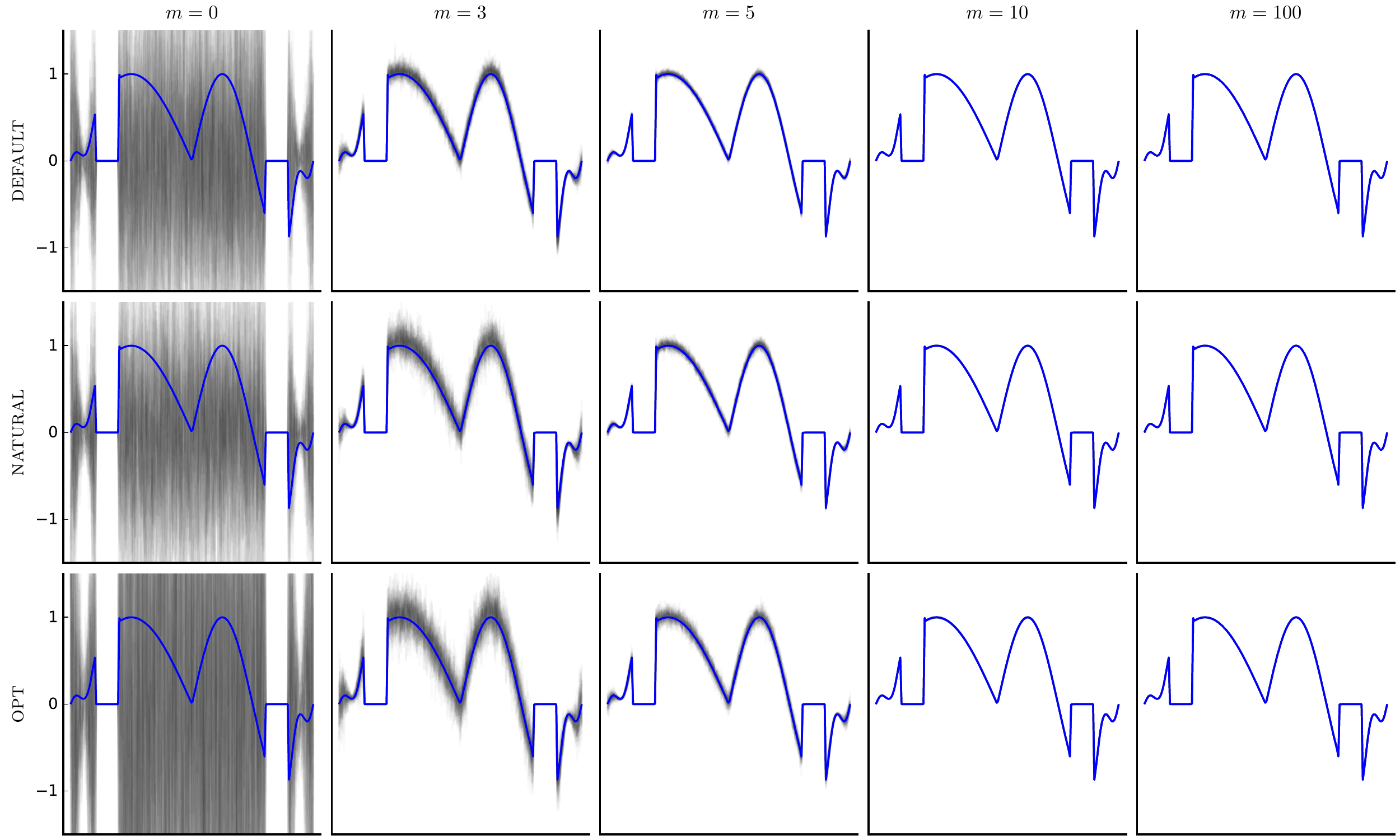}
	\caption{Samples from the distributional output of the probabilistic iterative method implied by using the conjugate gradient method as the underlying iterative method.} \label{fig:cg}

\end{figure}

\subsection{Testing Calibration} \label{sec:experiment_calibration}

We now test for evidence against weak calibration for all of the probabilistic iterative methods and initial distributions considered.
Recall that, according to the results in \cref{sec:calibration}, stationary Richardson iterations give rise to probabilistic iterative methods that are strongly calibrated when $\omega$ is fixed (irrespective of whether the optimal step size or a fixed step size is used).
Non-stationary Richardson iteration with adaptive step size $\omega_m$ is conjectured to also give rise to a probabilistic iterative method that is strongly calibrated, as is the higher order method described above, but these strong calibration results have not been established.
It is unknown whether probabilistic iterative methods based on CG are strongly or weakly calibrated.
In addition to probabilistic iterative methods, we also include BayesCG from \citet{Cockayne2019b}, which is not a probabilistic iterative method in the sense of this paper and is not expected to be strongly calibrated owing to the negative results presented in \citet{Cockayne2019b} and in \citet{Reid2020}.
It was hitherto unknown whether BayesCG is weakly calibrated.

To test the hypothesis that probabilistic iterative methods are weakly calibrated, we apply the MMD-based test described in \cref{sec:testing_weak}.
For each initial distribution and each iterative method we generated $N=100$ independent samples from $\mu_0$ and $\nu_m$ from which the test statistic \cref{eq:mmd_statistic} was computed.
Significance was assessed using the bootstrap method with $M=1000$.
The kernel $k$ used was the squared exponential kernel from \cref{eq:sqexp}, with the length-scale set using the \emph{median heuristic} as recommended in \citet{Gretton2012}.
For each method, $m=10$ iterations were performed.
For the second order method, we opted to use the \textsc{rich} initial distribution for $\mathbf{x}_1$.

\begin{table}
	\centering
	\resizebox{\textwidth}{!}{%
	\begin{tabular}{|c|c|c|c|c|c|c|c|}
	\hline
		&& Rich. & Rich. & Rich. & Rich. & &  \\
		&& (default) & (optimal) & (adaptive) & (2o) & CG & BayesCG \\\hline
		\multirow{2}{*}{\textsc{default}} 
& $\textsc{mmd}_{\mathcal{F}_k}^2$ & 1.90e-04 & -3.11e-05 & 9.76e-06 & 5.36e-05 & -2.80e-05 & \textbf{1.14e-03} \\
& $q$ & 0.34 & 0.52 & 0.45 & 0.43 & 0.49 & \textbf{0.03} \\\hline
		\multirow{2}{*}{\textsc{natural}} 
& $\textsc{mmd}_{\mathcal{F}_k}^2$ & -1.72e-04 & -2.71e-04 & -2.44e-04 & -3.20e-04 & -2.98e-04 & \textbf{4.18e-03} \\
& $q$ & 0.60 & 0.64 & 0.64 & 0.68 & 0.68 & \textbf{0.00} \\\hline
		\multirow{2}{*}{\textsc{opt}} 
& $\textsc{mmd}_{\mathcal{F}_k}^2$ & 3.59e-05 & 1.00e-05 & 4.30e-06 & -6.62e-06 & 3.57e-05 & \textbf{6.57e-03} \\
& $q$ & 0.48 & 0.47 & 0.48 & 0.49 & 0.47 & \textbf{0.00} \\\hline
	\end{tabular}%
	}
	\caption{
		Results from applying the maximum mean discrepancy (MMD)-based test from \cref{sec:testing_weak} to the methods described in \cref{sec:experiments}.
		The abbreviation ``Rich.'' refers to Richardson iteration. 
		``2o'' refers to the second order method.
		The test does not reject the null that each of the methods assessed is weakly calibrated, with the exception of BayesCG where the null is rejected.
		Results that are statistically significant at the 5\% level, indicating that the method is not weakly calibrated, are highlighted in bold.
	} \label{table:calib_results}
\end{table}

\cref{table:calib_results} shows test statistics obtained for each of these methods arising in the test for weak calibration described in \cref{sec:testing_weak}, for each choice of initial distribution from \cref{sec:prior}.
Reported are the value of \cref{eq:mmd_statistic} (as \textsc{mmd} in \cref{table:calib_results}).
Note that while strictly speaking \textsc{mmd} ought to be positive, due to sampling error it may be negative; this was also observed in \citet{Gretton2012}.
Also reported is the statistic $q$, which is analogous to a $p$-value in a classical statistical test; if $q'$ is the empirical quantile of \textsc{mmd} within the empirical distribution based on $M$ bootstrap samples of \cref{eq:mmd_statistic}, then $q = 1-q'$.
Thus, a small $q$ represents evidence that the PNM is not weakly calibrated.
We used the value $\alpha = 0.05$, representing a $5\%$ significance level, as a threshold in \cref{table:calib_results}; thus, if a value of $q$ below $0.05$ was obtained this constitutes evidence that the method is \emph{not} weakly calibrated.
Note that owing to the fact that $q$ is based on a sample from the bootstrapped distribution, it is possible to obtain $q=0$; we would expect the true $p$-value to be small but positive.

Examining the results, Richardson iteration with both default and optimal step sizes is seen to be weakly calibrated.
This provides support for our testing methodology, since from \citet[Lemma 2.19]{Oates2020} any strongly calibrated PNM must be weakly calibrated.
Similarly the second order method is weakly calibrated, which is to be expected since the proof of strong calibration for this method would require only a small extension relative to the case of a first order method.
Richardson iteration with the adaptive step size appears to be weakly calibrated for all initial distributions considered, suggesting that the non-stationarity implied by the adaptive step size does not affect the weak calibration of the method.

Perhaps more surprisingly, owing to its high degree of nonlinearity, CG also appears to be weakly calibrated.
This hints at the possibility of a more fundamental result regarding the calibration of probabilistic iterative methods in the general setting, though we leave study of this conjecture to future work.

Concerning BayesCG (which we emphasise again is \emph{not} a probabilistic iterative method in the same sense as the other methods considered), the results show that BayesCG is \emph{not} weakly calibrated for either the \textsc{natural} or \textsc{opt} initial distributions $\mu_0$ even when the prior distribution, required in BayesCG, is set equal to $\mu_0$ itself.
This is to be expected, considering that this method is known \emph{not} to produce meaningful posteriors apart from in special cases \citep[e.g.~][]{Reid2020}.
One other noteworthy point is that for the \textsc{default} initial distribution the \textsc{mmd} obtained for BayesCG has a slightly higher value of $q=0.03$.
This is perhaps due to the fact that, with such an uninformative prior, BayesCG is known to converge quite slowly.
Thus the posterior after $10$ iterations may not have deviated far from the prior.

\subsection{Spectral Behaviour} \label{sec:insights}
Lastly we examine the spectral behaviour of one of the methods above by performing a principal component analysis, to illustrate how the output of a probabilistic iterative method can provide a richer description of error compared to a classical error bound.
In this section we fixed the distribution $\mu_0$ to \textsc{natural}.

Here we consider principal components (leading eigenvectors) of the covariance matrix $\mat{A} \mat{\Sigma}_m \mat{A}^\top$, which describes covariance in the domain of the function \cref{eq:gs}.
The six leading principal components for the probabilistic iterative method based on Richardson iteration with default step size $\omega = 2/3$ are displayed in \Cref{fig:richardson_pc}.
At each of the values of $m$ considered, the low frequency variation over the interval $[0.2,0.8]$ is seen to be the dominant principal component (more so as $m$ is increased), which accords with the result of \Cref{fig:richardson_default} in that the error of \textsc{natural} is mainly manifest in a low-frequency vertical shift between the exact interpolant and the sampled output.
At $m=100$ the first six components account for over $60\%$ of the variability in the distributional output, with the remaining variability dedicated to higher-frequency aspects of the solution.

The detailed nature of these error indicators may be useful to shed light on the aspects of the exact solution $\xtrue $ that we are most uncertain about, having run a finite number of iterations of a probabilistic iterative method.
This rich description of numerical uncertainty can trivially be propagated through subsequent computation $F(\xtrue )$, e.g.\ by sampling from $\mu_m$ and then applying $F$, in order to probabilistically assess the impact of numerical uncertainty on any subsequent computational output.

\begin{figure}[t!]
\centering
\includegraphics[width = 0.9\textwidth]{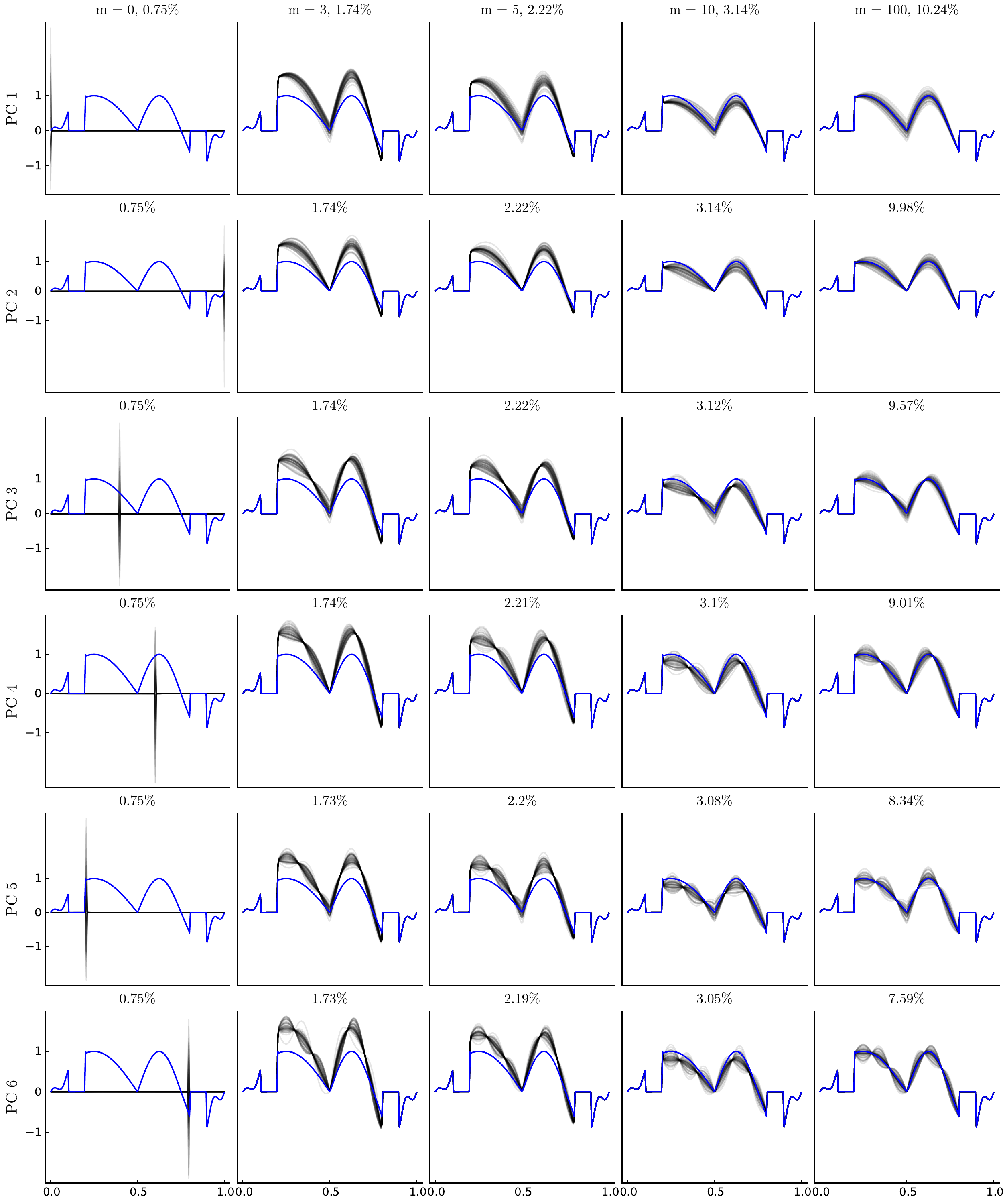}
\caption{A closer look at the distributional output: principal components from a probabilistic iterative method based on Richardson iteration with the default step size and initial distribution \textsc{natural}.
Here the first 6 principal components (PC) are displayed for the same values of $m$ used in \Cref{fig:richardson_stationary}.
The percentages indicate the percentage of the total variation that is explained by that component.
Each grey line is constructed as the mean of $\mu_m$, plus a sample in the direction of the relevant principal component, re-scaled to improve visualisation, with 50 samples shown in total.
}
\label{fig:richardson_pc}
\end{figure}

\section{Conclusion} \label{sec:conclusion}

In this paper we have introduced probabilistic iterative methods, a new class of probabilistic numerical methods for solving linear systems.
We have provided theoretical results concerning the convergence and calibration of these methods in the stationary and linear setting, and examined their empirical performance using a synthetic test-bed.
Finally, we alluded to how the output of a probabilistic iterative method could be used represent \textit{numerical uncertainty} and how such a representation could be propagated through subsequent computational output.

Several interesting avenues for future related work are now highlighted:

\subsection{Generalisation to Nonlinear Methods}

The generalisation of this work to nonlinear iterative methods, such as CG \citep{Hestenes1952} and other Krylov methods is of interest. 
These methods are more widely used than stationary iterative methods in modern applications, owing both to their faster convergence and that they only require access to the action of $\mat{A}$, rather than needing to interrogate and modify the elements of $\mat{A}$.

The definition that we proposed for probabilistic iterative methods in \cref{def:iterative_method}, and the sampling algorithm for accessing the output of a probabilistic iterative method described in \cref{sec:linear_probabilistic_methods}, do not require the generating iterative method to be linear.
However, with the exception of \cref{thm:iterative_contraction}, the theoretical results presented in this paper depend strongly on linearity.
The experimental results in \cref{sec:experiment_calibration} indicate that CG, a prototypical nonlinear iterative method, may be weakly calibrated.
The goal of theoretically establishing the calibration properties of nonlinear probabilistic iterative methods represents interesting future work.

\subsection{Gradient Flow Interpretation}

Recent work in the numerical analysis community highlights that iterative methods for linear systems may be interpreted as the discrete-time solution of an underlying dynamical system on $\mathbb{R}^d$ \citep{Chu2008}.
Insight may then be gained by studying the original dynamical system.
In parallel, recent work in the statistics and machine learning communities has provided gradient flow interpretations of various sampling and variational inference algorithms on $\mathcal{P}(\mathbb{R}^d)$ \citep[e.g.][]{Arbel2019,Liu2019}
An interesting avenue for future work would be to consider whether the methods presented in this paper may be interpreted as a discretisation of a gradient flow on $\mathcal{P}(\mathbb{R}^d)$, and whether insight can be gained by performing analysis of the continuous flow.

\subsection{Wider Applications}

In this paper we have focussed on iterative methods for solving linear systems.
However, the assumption that $\itermethod$ was an iterative method for solving such systems was not essential to \cref{def:iterative_method}.
Provided an initial distribution $\mu_0$ can be constructed in the domain of $\itermethod_\#$, probabilistic iterative methods could be applied to \emph{any} classical problem for which iterative methods are used, such as solvers for eigenproblems, numerical optimisation problems or even solvers for nonlinear differential equations.
\cref{thm:iterative_contraction} also applies to this general case, provided a suitable bound of the form in \cref{eq:generic_bound} can be derived in a norm adapted to the problem and, when the iteration is an affine map, we expect that the proof techniques from \cref{sec:calibration} could be applied.

\paragraph{Acknowledgements}
JC was supported by Wave 1 of the UKRI Strategic Priorities Fund under the EPSRC Grant EP/T001569/1, particularly the ``Digital Twins for Complex Engineering Systems'' theme within that grant, and the Alan Turing Institute.
The work of ICFI was supported in part by National Science Foundation
grants DMS-1760374 and DMS-1745654.
CJO was supported by the Lloyd's Register Foundation programme on data-centric engineering at the Alan Turing Institute, UK.
The work of TWR was supported in part by National Science Foundation
grant DMS-1745654.

\clearpage
\begin{appendices}
\crefalias{section}{appendix}

\section{Proof of \texorpdfstring{\cref{thm:calibration_diagonalisable}}{Theorem \ref{thm:calibration_diagonalisable}}} \label{supp:sec:proof_calibration_diagonalisable}

In order to prove \cref{thm:calibration_diagonalisable}, we need several results from linear algebra about the range and kernel of products of matrices, as well as decomposition of a diagonalizable matrix.

\begin{lemma}[{\citet[Fact 6.3]{Ipsen2009}}] 
  \label{L:KernelEqual}
  Let $\Ymat,\Wmat\in\reals^{d\times d}$. If $\Ymat$ is non-singular, then $\ker(\Ymat\Wmat) = \ker(\Wmat)$.
\end{lemma}

\begin{lemma}[{\citet[Facts 6.3 and 6.4]{Ipsen2009}}]
  \label{L:DiagRange}
  Let $\Ymat,\Omegamat,\Wmat\in\reals^{d\times d}$ where $\Ymat$ and $\Wmat$ are  non-singular and $\Omegamat$ is diagonal. 
  If $\Ymat$, $\Omegamat$, and $\Wmat$ have the partitions
  \begin{equation*}
    \Ymat = \begin{bmatrix} \Ymat_1 & \Ymat_2 \end{bmatrix}, \quad
    \Omegamat = \begin{bmatrix} \Omegamat_{11} & \zerovec \\ \zerovec & \zerovec \end{bmatrix}, \quad \text{and} \quad
    \Wmat = \begin{bmatrix} \Wmat_1^\top \\ \Wmat_2^\top \end{bmatrix},
  \end{equation*}
  with $\Ymat_1,\Wmat_1\in\reals^{d\times r}$, $\Ymat_2,\Wmat_2\in\reals^{d\times (d-r)}$, and $\Omegamat\in\reals^{r\times r}$, then
  \begin{equation*}
    \range(\Ymat\Omegamat\Wmat) = \range(\Ymat_1) \quad \text{and} \quad \ker(\Ymat\Omegamat\Wmat) = \range(\Wmat_2).
  \end{equation*}
\end{lemma}

\begin{lemma}{\citet[Lemma 3.4.1.10]{Horn2009}} \label{lemma:real_jordan}
	Let $\mat{G} \in \reals^{d\times d}$ be diagonalisable and of rank $r < d$.
	Then $\mat{G}$ may be represented in its real Jordan canonical form as
	\begin{equation*}
		\mat{G} = \mat{Y} \mat{\Omega} \mat{Y}^{-1}
	\end{equation*}
	where $\mat{Y} \in \reals^{d\times d}$ is invertible, while $\mat{\Omega} \in \reals^{d\times d}$ is of the form
	\begin{align*}
		\mat{\Omega} = \begin{pmatrix}
			\mat{\Omega}_{11} & \mat{0} \\
			\mat{0} & \mat{0}
		\end{pmatrix} .
	\end{align*}
	Here $\mat{\Omega}_{11} \in \reals^{r\times r}$ is nonsingular and block-diagonal, with $\ell$ $2\times 2$ blocks and $s$ $1\times 1$ blocks, where $\ell$ is the number of nonzero conjugate pairs of complex eigenvalues of $\mat{G}$ and $s$ is the number of nonzero real eigenvalues of $\mat{G}$, so that $r = 2\ell + s$.
\end{lemma}

\noindent With these results stated we proceed to the main proof:

\vspace{5pt}

\begin{proof}[Proof of \cref{thm:calibration_diagonalisable}]
  First note that if $\textup{rank}(\mat{G}) = d$ then $\mat{G}$ is invertible, so the probabilistic iterative method is strongly calibrated as a result of \cref{thm:uq_g_nonsingular}.
  Thus we focus on the case that $\textup{rank}(\mat{G}) < d$.

   We complete this proof in multiple steps:
   \begin{enumerate}[label=\textbf{Step \arabic*},left=0pt,ref={Step \arabic*}] 
   	\item We express the range and kernel of $\Sigmat_m$ in terms of the matrices forming the real Jordan canonical form of $\Gmat$, thus identifying the matrices $\mat{R}$ and $\mat{N}$ from \cref{thm:calibration_diagonalisable}.  \label{diag_proof:step_1}
   	\item We compute $(\mat{R}^\top\Sigmat_m\mat{R})^{1/2}$, $(\mat{R}^\top\Sigmat_m\mat{R})^{1/2}\mat{R}^\top(\xvec_m-\xtrue)$ and $\mat{N}^\top(\xvec_m-\xtrue)$.  \label{diag_proof:step_2}
    \item We combine these results to show that stationary iterative methods are strongly calibrated when $\mat{G}$ is diagonalisable. \label{diag_proof:step_3}
   \end{enumerate}
  
\paragraph{\ref{diag_proof:step_1}} 
We first compute the range and kernel of $\Sigmat_m$.
This covariance matrix is defined as
\begin{equation*}
  \Sigmat_m = \Gmat^m\Sigmat_0(\Gmat^m)^\top.
\end{equation*}
From \cref{lemma:real_jordan} we have that
\begin{equation*}
  \Gmat^i = \Ymat\Omegamat^i\Ymat^{-1}, \qquad 0\leq i \leq m.
\end{equation*}
We partition the diagonalization of $\Gmat$ as
\begin{equation*}
  \Ymat = \begin{bmatrix} \Ymat_1 & \Ymat_2 \end{bmatrix}, \quad \Omegamat = \begin{bmatrix} \Omegamat_{11} & \zerovec \\ \zerovec & \zerovec \end{bmatrix}, \quad \text{and} \quad \Ymat^{-1} = \begin{bmatrix} \Wmat_1^\top \\ \Wmat_2^\top \end{bmatrix},
\end{equation*}
where $\Ymat_1,\Wmat_1\in\reals^{d\times r}$, $\Ymat_2,\Wmat_2\in\reals^{d\times(d-r)}$, and $\Omegamat_{11}\in\reals^{r\times r}$. 
With this partitioning and \cref{L:DiagRange} we have
\begin{equation}
  \label{Eq:Grange}
  \range(\Gmat^i) = \range(\Ymat_1) \quad \text{and} \quad \ker((\Gmat^i)^\top) = \range(\Wmat_2), \qquad 0\leq i \leq m.
\end{equation}
We now express the range and kernel of $\Sigmat_m$ in terms of $\mat{Y}_1$ and $\mat{W}_2$. 
Express $\Sigmat_m$ as the product $\Sigmat_m = \Qmat\Qmat^\top$, where $\Qmat = \Gmat^m\Sigmat_0^{1/2}$. For any $\xdummyvec \in\ker(\Sigmat_m)$ we have 
\begin{equation*}
  \Sigmat_m\xdummyvec = \zerovec \iff \xdummyvec^\top\Sigmat_m\xdummyvec = \zerovec \iff (\Qmat^\top\xdummyvec)^\top\Qmat^\top\xdummyvec = \zerovec \iff \Qmat^\top\xdummyvec = \zerovec.
\end{equation*}
Thus $\ker(\Sigmat_m) = \ker(\Qmat^\top)$. Because $\Sigmat_0^{1/2}$ is the non-singular square root of the non-singular matrix $\Sigmat_0$, we can apply \cref{L:KernelEqual}  to $\Qmat^\top = \Sigmat_0^{1/2}(\Gmat^m)^\top$ to obtain
\begin{equation}
  \label{Eq:SigKer}
  \ker(\Sigmat_m) = \ker(\underbrace{\Sigmat_0^{1/2}(\Gmat^m)^\top}_{\Qmat^\top}) = \ker((\Gmat^m)^\top).
\end{equation}
By the fundamental theorem of linear algebra, $\ker((\Gmat^m)^\top)$ is the orthogonal complement of $\range(\Gmat^m)$ and $\ker(\Sigmat_m)$ is the orthogonal complement of $ \range(\Sigmat_m^\top) = \range(\Sigmat_m)$. This combined with \cref{Eq:SigKer} implies
\begin{equation}
  \label{Eq:SigRange}
  \range(\Sigmat_m)= \range(\Gmat^m).
\end{equation}
Applying \cref{L:DiagRange} with $\mat{W} = \mat{Y}^{-1}$ gives
\begin{equation}
  \label{Eq:SigRangeKer}
  \range(\Sigmat_m) = \range(\Ymat_1) \quad \text{and} \quad \ker(\Sigmat_m) = \range(\Wmat_2).
\end{equation}
Therefore, referring to \cref{thm:calibration_diagonalisable}, we have that $\mat{R} = \Ymat_1$ and $\mat{N} = \Wmat_2$.

\paragraph{\ref{diag_proof:step_2}}
We begin by computing $(\Ymat_1^\top\Sigmat_m\Ymat_1)^{1/2}$.
We have that
\begin{align*}
  \Ymat_1^\top\Sigmat_m\Ymat_1 &= \Ymat_1^\top\Gmat^m\Sigmat_0(\Gmat^m)^\top \Ymat_1 \\
  &= \Ymat_1^\top\Ymat\Omegamat^m\Ymat^{-1}\Sigmat_0\Ymat^{-\top}(\Omegamat^m)^\top\Ymat^\top\Ymat_1 \\
  &= \Ymat_1^\top\Ymat_1\Omegamat_{11}^m\Wmat_1^\top \Sigmat_0\Wmat_1(\Omegamat_{11}^m)^\top\Ymat_1^\top\Ymat_1.
\end{align*}
The product $\Ymat_1^\top\Ymat_1$ is Hermitian positive definite because $\Ymat_1$ is full rank. 
Additionally, $\Ymat_1\Wmat_1^\top = \Imat_r$ because $\Ymat\Ymat^{-1} = \Imat_n$. 
Therefore the inverse square root\footnote{This is a square root in the sense of \cref{sec:notation}, a matrix $\Tmat^{1/2}$ such that $\Tmat^{1/2} (\Tmat^{1/2})^\top = \Tmat$.} is,
\begin{equation}
  \label{Eq:PostProj}
  (\Ymat_1^\top\Sigmat_m\Ymat_1)^{-1/2} =\Bmat\Omegamat_{11}^{-m}(\Ymat_1^\top\Ymat_1)^{-1},
\end{equation}
where $\Bmat = (\Wmat_1^\top\Sigmat_0 \Wmat_1)^{-1/2} \in\reals^{r\times r}$.

Next, we compute $(\Ymat_1^\top\Sigmat_m\Ymat_1)^{1/2}\Ymat_1^\top(\xtrue - \xvec_m)$.
Left-multiplying $\xtrue - \xvec_m$ by $\Ymat_1^\top$ yields 
\begin{align}
  \Ymat_1^\top(\xtrue - \xvec_m) &= \Ymat_1^\top \left(\xtrue - \Gmat^m \xvec_0 - \sum_{i=0}^{m-1} \Gmat^i \mathbf{f} \right) \\
  &= \Ymat_1^\top\left( \xtrue - \Ymat\Omegamat^m\Ymat^{-1}\xvec_0 - \mathbf{f} - \sum_{i=1}^{m-1} \Ymat\Omegamat^i\Ymat^{-1}\mathbf{f} \right) \notag\\
  &= \Ymat_1^\top\xtrue - \Ymat_1^\top\Ymat_1\Omegamat_{11}^m\Wmat_1^\top\xvec_0 - \Ymat_1^\top \mathbf{f} - \sum_{i=1}^{m-1} \Ymat_1^\top\Ymat_1\Omegamat_{11}^i\Wmat_1^\top \mathbf{f}. \label{Eq:MeanProj}
\end{align}
Now left-multiplying by \cref{Eq:PostProj} gives
\begin{multline}
  \label{Eq:MeanProj2}
  (\Ymat_1^\top\Sigmat_m\Ymat_1)^{-1/2}\Ymat_1^\top(\xtrue - \xvec_m) \\
  = \underbrace{\Bmat\Omegamat_{11}^{-m}(\Ymat_1^\top\Ymat_1)^{-1}\left(\Ymat_1^\top(\xtrue - \mathbf{f}) - \sum_{i=1}^{m-1}\Ymat_1^\top\Ymat_1\Omegamat_{11}^i\Wmat_1^\top\mathbf{f}\right)}_{(\star)} - \Bmat\Wmat_1^\top\xvec_0.
\end{multline}
We now focus on simplifying ($\star$). 
Left-multiplying \cref{eq:fixed_point_eq} by $\Ymat_1^\top$ gives
\begin{align*}
  \Ymat_1^\top\xtrue &= \Ymat_1^\top\Ymat_1\Omegamat_{11}\Wmat_1^\top\xtrue+\Ymat_1^\top\mathbf{f}. \\
  \implies \Wmat_1^\top \xtrue &= \Omegamat_{11}^{-1}(\Ymat_1^\top\Ymat_1)^{-1}\Ymat_1^\top(\xtrue - \mathbf{f})\numberthis \label{Eq:FixedY} 
\end{align*}
while left-multiplying by $\Wmat_1^\top$ gives
\begin{align*}
  \Wmat_1^\top\xtrue &= \Omegamat_{11}\Wmat_1^\top\xtrue+\Wmat_1^\top\mathbf{f} \\
  \implies \Wmat_1^\top \xtrue &= \Omegamat_{11}^{-1}\Wmat_1^\top(\xtrue-\mathbf{f}) .\numberthis \label{Eq:FixedX}
\end{align*}
Substituting \cref{Eq:FixedY} into ($\star$) results in
\begin{align*}
  (\star) &= \Bmat\Omegamat_{11}^{-m}(\Ymat_1^\top\Ymat_1)^{-1}\left(\Ymat_1^\top(\xtrue - \mathbf{f}) - \sum_{i=1}^{m-1}\Ymat_1^\top\Ymat_1\Omegamat_{11}^i\Wmat_1^\top\mathbf{f}\right) \\
  &= \Bmat\Omegamat_{11}^{-(m-1)}\left(\Omegamat_{11}^{-1}(\Ymat_1^\top\Ymat_1)^{-1}\Ymat_1^\top(\xtrue - \mathbf{f}) - \Omegamat_{11}^{-1}(\Ymat_1^\top\Ymat_1)^{-1}\sum_{i=1}^{m-1}\Ymat_1^\top\Ymat_1\Omegamat_{11}^i\Wmat_1^\top\mathbf{f}\right) \\
  &= \Bmat\Omegamat_{11}^{-(m-1)}\left(\Wmat_1^\top\xtrue - \Wmat_1^\top\mathbf{f} - \sum_{i=1}^{m-2}\Omegamat_{11}^i\Wmat_1^\top\mathbf{f}\right).
\end{align*}
Repeatedly substituting \cref{Eq:FixedX} into the previous equation gives
\begin{align*}
  (\star) &= \Bmat\Omegamat_{11}^{-(m-1)}\left(\Wmat_1^\top(\xtrue - \mathbf{f}) - \sum_{i=1}^{m-2}\Omegamat_{11}^i\Wmat_1^\top\mathbf{f}\right) \\
          &= \Bmat\Omegamat_{11}^{-(m-2)}\left(\Omegamat_{11}^{-1}\Wmat_1^\top(\xtrue-\mathbf{f}) - \Omegamat_{11}^{-1}\sum_{i=1}^{m-2} \Omegamat_{11}^i\Wmat_1^\top\mathbf{f}\right) \\
          &= \Bmat\Omegamat_{11}^{-(m-2)}\left(\Wmat_1^\top(\xtrue - \mathbf{f}) - \sum_{i=1}^{m-3}\Wmat_1^\top\mathbf{f}\right) \\
          &\phantom{=} \  \vdots \\
          &= \Bmat\left(\Omegamat_{11}^{-1}\Wmat_1^\top(\xtrue-\mathbf{f})\right) \\
          &= \Bmat\Wmat_1^\top\xtrue.
\end{align*}
Finally substituting this back into \cref{Eq:MeanProj2} shows
\begin{equation}
  \label{Eq:ErrorProjRange}
  (\Ymat_1^\top\Sigmat_m\Ymat_1)^{-1/2}\Ymat_1^\top(\xtrue  - \xvec_m)  = \Bmat \Wmat_1^\top(\xtrue  - \xvec_0).
\end{equation}

Lastly we compute $\Wmat_2^\top(\xtrue  - \xvec_m)$.
This follows a similar argument to the above. 
We have
\begin{equation} \label{eq:fixed_point_W2}
  \Wmat_2^\top(\xtrue  - \xvec_m) = \Wmat_2^\top(\xtrue  - \mathbf{f})
\end{equation}
since $\Wmat_2^\top\Gmat = \zerovec$. 
Similarly, left-multiplying the fixed-point equation \cref{eq:fixed_point_eq} by $\Wmat_2^\top$ gives
\begin{equation*}
  \Wmat_2^\top\xtrue  = \Wmat_2^\top\mathbf{f}
\end{equation*}
Substituting this into \cref{eq:fixed_point_W2} gives
\begin{equation}
  \label{Eq:ErrorProjNull}
  \Wmat_2^\top(\xtrue  - \xvec_m) = \zerovec .
\end{equation}

\paragraph{\ref{diag_proof:step_3}} 
\cref{Eq:ErrorProjNull} validates the second requirement of \cref{def:strong_pim_singular}, since $\mat{N} = \Wmat_2$.
It remains to establish the first requirement.
To accomplish this replace $\xtrue $ with $\bm{X} \sim\N(\xvec_0,\Sigmat_0)$ in \cref{Eq:ErrorProjRange}.
Since $\Wmat_1^\top \bm{X} \sim \N(\Wmat_1^\top\xvec_0, \Wmat_1^\top\Sigmat_0 \Wmat_1^\top)$, it follows that
\begin{equation*}
  \Bmat \Wmat_1^\top(\bm{X} - \xvec_0) = (\Wmat_1^\top \Sigmat_0 \Wmat_1^\top)^{-\frac{1}{2}} \Wmat_1^\top (\bm{X} - \xvec_0) \sim \N(\zerovec,\Imat_r).
\end{equation*}
which verifies the first requirement and completes the proof.
\end{proof}
\end{appendices}

\bibliographystyle{\thebibstyle}
\bibliography{refs}

\end{document}